\documentclass[twocolumn]{autart}    

\usepackage{bbold} 

\usepackage{amsmath}
\usepackage{color}
\usepackage{algorithm}
\usepackage{cite}
\usepackage{mathtools}
\usepackage{tikz}
\usepackage{appendix}

\usepackage{tikz}
\usetikzlibrary{arrows.meta} 

\DeclareMathOperator{\reduce}{{\text{Red}}}
\DeclareMathOperator{\diag}{{\text{Diag}}}

\newtheorem{subproblem}{Subproblem}
\newtheorem{definition}{Definition}

\newtheorem{problem}{Problem}
\newtheorem{lemma}{Lemma}
\newtheorem{proof}{Proof}
\newtheorem{remark}{Remark}
\newtheorem{assumption}{Assumption}
\newtheorem{theorem}{Theorem}

\DeclareMathOperator{\dist}{{\text{dist}}}

\begin{document}

\begin{frontmatter}
\title{Compositional Synthesis for Linear Systems via Convex Optimization of Assume-Guarantee Contracts} 


\author[kasra]{Kasra Ghasemi}\ead{kasra0gh@bu.edu},    
\author[sadra]{Sadra Sadraddini}\ead{sadra@mit.edu},              
\author[kasra]{Calin Belta}\ead{cbelta@bu.edu}  

\address[kasra]{Boston University, Boston, MA}  
\address[sadra]{Massachusetts Institute of Technology, Cambridge, MA}             


\begin{abstract}  
We take a divide and conquer approach to design controllers for reachability problems given large-scale linear systems with polyhedral constraints on states, controls, and disturbances. Such systems are made of small subsystems with coupled dynamics. We treat the couplings as additional disturbances and use assume-guarantee (AG) contracts to characterize these disturbance sets. For each subsystem, we design and implement a robust controller locally, subject to its own constraints and contracts. The main contribution of this paper is a method to derive the contracts via a novel parameterization and a corresponding potential function that characterizes the distance to the correct composition of controllers and contracts, where all contracts are held. We show that the potential function is convex in the contract parameters. This enables the subsystems to negotiate the contracts with the gradient information from the dual of their local synthesis optimization problems in a distributed way, facilitating compositional control synthesis that scales to large systems. We present numerical examples, including a scalability study on a system with tens of thousands of dimensions, and a case study on applying our method to a distributed Model Predictive Control (MPC) problem in a power system.  
\end{abstract}
\end{frontmatter}
\section{Introduction}
Verification and synthesis of controllers for large-scale systems with hard constraints are encountered in many applications such as resource management, biology, and robotics \cite{bakule2008decentralized, antonelli2013interconnected}. The large-scale nature of such systems poses certain computational challenges for analysis and controller design. First, verification and synthesis problems usually reduce to  optimization problems with prohibitive sizes. This particularly concerns the applicability of real-time algorithms such as Model Predictive Control (MPC). Second, large-scale systems are typically networks of many compartments with physically separated actuators. A central command unit would lead to a complex communication architecture during implementation, commonly referred to as online mode. Thus decentralized computations of the controllers are preferred such that each subsystem computes its own decisions based on local information.

Roughly, decentralized control can be categorized into two approaches. In the first, the synthesis problem is solved in a centralized fashion, but the controllers are deployed in a decentralized manner in the online mode. In this case, one has to pose structural constraints and costs on the sparsity of the control law and solve a large optimization problem subject to those constraints (in offline mode) before the deployment phase. The works in \cite{rotkowitz2005characterization,fardad2011sparsity,lin2013design,nilsson2016synthesis} follow this approach. In the second, the synthesis problem itself is solved in a distributed manner. In other words, the control laws are determined by iteratively solving local synthesis problems. In this case, there is no single large optimization problem, and the corresponding computation time is significantly reduced. 

In this paper, we are interested in the latter case, which is more challenging, as one has to formulate and solve a distributed optimization problem. We consider linear discrete-time systems with additive set-valued bounded disturbances and constrained states and control inputs. Couplings across subsystems are allowed over states, controls, and constraints. The problem is to design decentralized feedback controllers such that all closed-loop trajectories of the aggregate system satisfy all the constraints under all disturbances, while a cost function is optimized. The cost can be chosen trivially as a constant if the goal is to only find feasible controllers. If this is not the case, we assume that the cost function is separable over the subsystems. Both finite-time (with possibly time-varying system dynamics and constraints) and infinite-time (in the form of robust set-invariance) problems are considered. The subsystems can communicate offline to synthesize the decentralized controllers and identify their closed-loop set-valued trajectories, also known as reachability sets. Thus, online computations use only local information. If online communication is possible, controllers and set-valued trajectories can be recomputed online in an MPC fashion with recursive feasibility guarantee. 

The system interacts with its external world, which we refer to as its environment. Exogenous disturbances are examples of such interactions. Central to our approach is a divide and conquer approach based on assume-guarantee (AG) reasoning \cite{henzinger1998you}. We use the notion of AG contracts to formalize the promises a system takes from and makes to its environment. One use of AG contracts is to decouple subsystems within a network by constructing one contract for each subsystem. By doing so, we are seeing the effect of couplings as exogenous, which can lead to mismatch between the assumptions over the environment and what happens in reality to some subsystems. Correct composition happens when there is no mismatch in a set of AG contracts. Achieving correct composition is a challenging problem. While AG contracts have been used in control synthesis \cite{henzinger2001assume,nuzzo2013contract,bogomolov2014assume}, the focus has been on how to use AG contracts given a priori. In contrast, the focus of this paper is on synthesizing contracts. The main contributions of this paper are as follows:

\begin{itemize}
    \item We provide a novel parameterization of AG contracts in the form of zonotopes and a corresponding potential function that quantifies how far a set of contracts is from correct composition. We show that the potential function and the set of parameters leading to correct composition are convex (Section \ref{sec_AG} and Section \ref{sec_AGR}).
    
    \item We develop an algorithm for subsystems to coordinate AG contracts and solve local optimization problems such that all subsystems' closed-loop responses are guaranteed to satisfy the constraints over states and control inputs. We also show that our distributed parameter synthesis method is not more conservative than the centralized version of contract optimization (Section \ref{sec_gradient}). 

    \item We include examples showing that our method can handle in a reasonable amount of time systems of orders of magnitude larger than what centralized methods can handle. We also provide an application to MPC for a power network (Section \ref{section:mpc} and Section \ref{sec_example}).  
\end{itemize}

\subsection{Related Work}
\label{sec_related}
\textbf{Distributed MPC~~~}
Distributed optimization for MPC has a long history. Earlier works such as \cite{summers2012distributed, trodden2013cooperative} focus on systems with decoupled dynamics. This way, the optimization problem needs to be only decomposed through the cost function. The work in \cite{conte2016distributed} handles coupled dynamics and formulates MPC as a distributed optimization problem that can be solved using methods such as Alternating Direction Method of Multipliers (ADMM) \cite{boyd2011distributed}. However, the authors of \cite{conte2016distributed} consider deterministic systems only and there are no set-valued couplings. In this paper, we deal with set-valued uncertainties and couplings through assume-guarantee contracts. 

\textbf{Contracts for Control~~~}
As mentioned earlier, the majority of works assume contracts given a priori and focus on how to use them for control synthesis \cite{nuzzo2013contract, saoud2020contract}. In \cite{scialanga2018robust}, the whole feasible sets of the individual subsystems were considered in the contracts. This approach is very conservative and the problem becomes infeasible in case the feasible sets are unbounded as it makes the couplings also unbounded. The work in \cite{Kim} studied parametric assume-guarantee contracts and developed a notion of small gain theorem for systems with dynamical couplings. In  \cite{sadraddini2017formal, darivianakis2018decentralized, Lin2020} the contracts were synthesized in a centralized optimization problem, but controllers were deployed in a decentralized way. The work in \cite{kim2015compositional} found contracts through binary search over box-shaped contract sets for monotone systems and specifications. The parameterization of contracts in \cite{Lin2020} leads to a non-convex optimization problem that is conservatively approximated by a convex semi definite program. In contrast to the works mentioned above, 
contract synthesis as well as the control synthesis are carried out in a distributed way in our paper. In addition, our parameterization has convexity properties, which combined with the distributability mentioned above, lead to scalability and convergence of our proposed method over a finite number of iterations.

\textbf{Controller Parameterization} Synthesizing controllers using convex programs is a well studied subject \cite{boyd2011distributed}. A seminal result is that linear feedback over past disturbances instead of the current state leads to a convex parameterization of system's predicted behaviours, which facilitates optimization subject to convex constraints \cite{goulart2006optimization}. The same idea can be applied to an output feedback setting \cite{anderson2019system}. The authors in \cite{Rakovic2007} also use a similar technique but for infinite time set invariance. The zonotope-based controllers in this paper are loosely based on disturbance feedback policies. By leveraging the properties of zonotopes and of zonotope operations, such as zonotope order reduction, we are able to provide a distributed set-based computation of robust controllers subject to bounded disturbance sets, which is correct-by-construction and scalable. 

\textbf{Conference Version~~~}
Some of the results from this paper appeared in the conference version \cite{myhscc}, where we introduced the convex parameterization and the synthesis method. The parameterization defined here is richer than the one in \cite{myhscc} 
due to the inclusion of contract set centroids, which is particularly useful for reachability problems. This paper also considers optimality in control synthesis as part of MPC, while \cite{myhscc} only focused on reach and avoid constraints.  In this paper, we also treat coupled constraints and provide a semidefinite programming approach to decouple them with the least amount of conservativesness.  Finally, complete proofs, updated examples, and an MPC case study are included in addition to the material from \cite{myhscc}.

\section{Preliminaries and Definitions}
\label{sec_prelim}

\subsection{Notation}
The sets of real, non-negative real, integers, and non-negative integers are represented by $\mathbb{R}$, $\mathbb{R}_+$, $\mathbb{N}$, and $\mathbb{N}_+$, respectively. We use $\mathbb{N}_{h_1,h_2}:= \{h_1, h_1+1,\cdots,h_2\}$, $h_1,h_2 \in \mathbb{N}, h_1 < h_2$, and $\mathbb{N}_h:=\mathbb{N}_{0:h}$. Additionally, $I_n$, $0_n$, and $\mathbb{1}_n$ represent the $n$-dimensional identity matrix, vector of zeros, and vector of ones, respectively. Given matrices $A_1, A_2, \cdots, A_N$ with the same number of rows, $[A_1,A_2, \cdots, A_N]$ is their horizontal concatenation. Given a matrix $A$, $|A|$ is its  element-wise absolute value. We use $\text{Blk}(.)$ to denote block diagonal concatenation of a matrix and $\text{sum()}$ to denote the sum of all its elements. Given $\alpha \in \mathbb{R}^n$, $\diag(\alpha)$ is the diagonal matrix composed of the entries of $\alpha$. The transpose of matrix $A$ is denoted by $A^T$. Inequality relations denoted by $\ge$ and $\le$ over matrices are interpreted element-wise, and $A \succeq 0$ indicates that $A$ is a positive semi-definite matrix. The infinity-norm of matrix $A \in \mathbb{R}^{n_1 \times n_2}, n_1 , n_2 \in \mathbb{N}_+$ is defined as  $\|A\|_\infty:=\text{max}( |A|\mathbb{1}_{n_1})$, where the max function returns the largest element.

Given $\mathbb{S}_1, \mathbb{S}_2 \subseteq \mathbb{R}^n$, their \emph{Minkowski sum} is denoted by $\mathbb{S}_1 \oplus \mathbb{S}_2 := \{s_1+s_2| s_1 \in \mathbb{S}_1, s_2 \in \mathbb{S}_2\}$. We interpret $s+\mathbb{S}$ as $\{s\} \oplus \mathbb{S}$, and $\mu(\mathbb{S})$ as $\{\mu(s)|s \in \mathbb{S}\}$, where $\mu$ is a function with $\mathbb{S}$ within its domain. The \textit{Directed Hausdorff distance} $d_{DH}(\mathbb{S}_1,\mathbb{S}_2)$ quantifies how distant is $\mathbb{S}_2$ from being a subset of $\mathbb{S}_1 $:
\begin{equation} \label{directed_huasdorff_distance_def}
    d_{DH}(\mathbb{S}_1,\mathbb{S}_2) :=  \sup_{s_2\in \mathbb{S}_2} \inf_{s_1\in \mathbb{S}_1} d(s_1,s_2) ,
\end{equation}
where $d: \mathbb{R}^n \times \mathbb{R}^n \rightarrow \mathbb{R}_+$ is a metric. For closed compact sets, $d_{DH}(\mathbb{S}_1,\mathbb{S}_2)=0$, if and only if $\mathbb{S}_2 \subseteq \mathbb{S}_1$. The Cartesian product of $\mathbb{S}_1$ and $\mathbb{S}_2$ is denoted by $
    \mathbb{S}_1 \times \mathbb{S}_2$  and the Cartesian product of $\mathbb{S}_1, \cdots, \mathbb{S}_N$ is denoted by $\prod_{i=1}^N \mathbb{S}_i$.  

\subsection{Zonotopes}
A zonotope $\mathcal{Z}(c,G)$ is defined as $ c \oplus G \mathbb{B}_p \subset \mathbb{R}^n$, where $c \in \mathbb{R}^n$ is the \emph{center}, the columns of $G \in \mathbb{R}^{n \times p}$ are the \emph{generators}, and $\mathbb{B}_p:=\{b\in \mathbb{R}^p| ||b||_\infty \leq 1\}$. The order of a zonotope is defined as $\frac{p}{n}$.
Zonotopes are convenient to manipulate with affine transformations and Minkowski sums:
\begin{subequations} \label{zonotopchar}
\begin{equation}
\label{eq_zonotope_affine}
    A \mathcal{Z}(c,G)+b=\mathcal{Z}(Ac + b,AG),
\end{equation}
\begin{equation}
\label{eq_zonotope_minkowski}
    \mathcal{Z}(c_1,G_1) \oplus \mathcal{Z}(c_2,G_2) = \mathcal{Z}(c_1+c_2,[G_1,G_2]).
\end{equation}
\end{subequations}
The volume of a full-dimensional zonotope of order 1, also known as parallelotope, is given as 
\cite{Gover2010}:
\begin{equation}\label{zonotope_volume}
    \text{Vol}(\mathcal{Z}(c,G)) = 2^{n} \text{det}(G^TG).
\end{equation}
The Cartesian product of zonotopes is:
\begin{equation*}
    \prod_{i=1}^N \mathcal{Z}\left(c_i,G_i) = \mathcal{Z}([c_1^T, 
    \cdots, c_N^T]^T,\text{Blk}(G_1,\cdots,G_N)\right).
\end{equation*}
Zonotope order reduction methods over-approximate a zonotope by another one with a smaller order. Several order reduction methods were reported in \cite{Kopetzki2018} and \cite{Yang2018}. In this paper, the \textit{Boxing method} \cite{C.Combastel,Kiihn1998} with order 1 is used to over-approximate a given zonotope $\mathcal{Z}(c,G)$ by a hyper-box:
\begin{equation} \label{reduction}
    \reduce(\mathcal{Z}(c,G)):=\mathcal{Z}(c,\diag(\sum_{i}{|g_i|})),
\end{equation}
where $g_i$ represents the $i$th column of $G$. Equation \eqref{reduction} is based on the following properties of zonotopes:
\begin{equation} \label{eq:box}
    c - \sum_{i=1}^p{|g_i|} \leq \mathcal{Z}(c,G) \leq c + \sum_{i=1}^p{|g_i|}.
\end{equation}
\begin{lemma}[Zonotope Containment \cite{sadraddini2019linear}]
\label{sadra_zon_containment}
Given two zonotopes $\mathcal{Z}(c_1,G_1)$ and $\mathcal{Z}(c_2,G_2)$, where $c_1 ,c_2 \in \mathbb{R}^q$ and $G_1 \in \mathbb{R}^{q \times r}$, $G_2 \in \mathbb{R}^{q \times s}$, we have $\mathcal{Z}(c_1,G_1) \subseteq \mathcal{Z}(c_2,G_2)$, if $\exists \Gamma \in \mathbb{R}^{s \times r}$ and $\gamma \in \mathbb{R}^s$ s.t.
\begin{subequations} \label{zon_containment}
\begin{equation}
    G_1 = G_2 \Gamma, 
\end{equation}
\begin{equation}
     c_2 - c_1 = G_2 \gamma, 
\end{equation}
\begin{equation} \label{con_zon_containment}
    \left\| [\Gamma,\gamma] \right\|_\infty \leq 1.
\end{equation}
\end{subequations}
\end{lemma}
While Lemma \ref{sadra_zon_containment} provides a sufficiency condition, it was shown in \cite{sadraddini2019linear} that its necessity gap, in a particular quantitative sense, is often small. The following modified version of Lemma \ref{sadra_zon_containment} enables us to independently scale each generator of 
$G_2$.
\begin{lemma}[Weighted Zonotope Containment] \label{lemma_weighted}
Given two zonotopes $\mathcal{Z}(\bar{c}_1,G_1)$ and $\mathcal{Z}(\bar{c}_2,G_2)$ (in the same format as Lemma \ref{sadra_zon_containment}), and a vector $\alpha \in \mathbb{R}^s_+$, we have  $\mathcal{Z}(\bar{c}_1,G_1) \subseteq \mathcal{Z}(\bar{c}_2,G_2\diag(\alpha))$, if the conditions in Lemma \ref{sadra_zon_containment} hold while constraint (\ref{con_zon_containment}) changes to the following element-wise inequality:
\begin{equation} \label{con_zon_addmod}
       [|\Gamma|,|\gamma|]  \mathbb{1}_s \leq \alpha.
\end{equation}
\end{lemma}
{\bf Proof:} see Appendix

\begin{lemma}\label{hasdurff_distance_computation}\textbf{(Directed Hausdorff Distance Computation \cite{sadraddini2019linear})}
Given two closed sets $\mathbb{S}_1$ and $\mathbb{S}_2$ and $\|.\|_\infty$ as the underlying metric, $d_{DH}(\mathbb{S}_1,\mathbb{S}_2)$ is the optimal value for the following linear program:
\begin{subequations} \label{directed_hausdorff_distance_computation}
\begin{align}
    \underset{d}{\text{argmin}} \quad & d \\
    \textrm{s.t.}  \quad & \mathbb{S}_2 \subseteq \mathbb{S}_1 \oplus d \mathcal{Z}(0_n,I_n) , \\
    & d \geq 0. 
\end{align}
\end{subequations}
\end{lemma}

\subsection{Constrained Linear Discrete-Time Systems}
A discrete-time (time-varying) linear system is characterized by the following equation:
\begin{equation}\label{singlesystem_LTV}
    x_{t+1} = A_tx_t+B_t u_t+d_t,
\end{equation}
where $x_t \in \mathbb{X}_t \subseteq \mathbb{R}^n$, $u_t \in \mathbb{U}_t \subseteq \mathbb{R}^m$, and $d_t \in \mathbb{D}_t \subseteq \mathbb{R}^n$
are the state, control, and disturbance at time $t \in \mathbb{N}$, respectively. The polytopic sets $\mathbb{X}_t$, $\mathbb{U}_t$, and $\mathbb{D}_t$ define time varying bounds and are assumed to be given. The matrices $A_t \in \mathbb{R}^{n\times n}$ and $B_t \in \mathbb{R}^{n \times m}$ may be time dependent. 
If $A_t$, $B_t$ and the sets $\mathbb{X}_t , \mathbb{U}_t$, and $\,\mathbb{D}_t$ are all time-invariant, the system is linear time-invariant (LTI). If at least one of them is time-variant, the system is linear time-variant (LTV) and we are mainly interested in a time-limited response over a finite horizon $h \in \mathbb{N}_+$ (i.e. $t\in \mathbb{N}_h$). Due to space limits, the optimization problems and algorithms shown here are for the LTV class of problems. The LTI versions can be seen as particular cases and can be formulated easily. 

A control policy $\mu$ is characterized by a set of functions $\mu_t : \mathbb{X}_t \rightarrow \mathbb{U}_t, t \in \mathbb{N}_h$, where $h \in \mathbb{N}_+$ is a finite horizon. For infinite horizon, the policy $\mu(.): \mathbb{X} \rightarrow \mathbb{U}$ is not a function of time. 
\begin{definition}[Finite-time Viable Sets] \label{def:viable}
For a given horizon $h \in \mathbb{N}_+$, a sequence of sets $\Omega_0,\Omega_1, ... , \Omega_h$ for system \eqref{singlesystem_LTV} is a sequence of \emph{viable} sets, if for all $t \in \mathbb{N}_h$, $\Omega_t \subseteq X_t$ and there exists a policy $\mu_t$ such that $\Theta_t \subseteq U_t$, where $\Theta_t :=\mu_t( \Omega_t)$, and 
\begin{equation} \label{eq:4}
    \forall t \in \mathbb{N}_{h-1}, \forall x_t\in \Omega_t, \forall d_t\in D_t   \Rightarrow x_{t+1}\in \Omega_{t+1}.
\end{equation}
$\Theta_t$ is called action set.
\end{definition}
\begin{definition}[Infinite-time Viable Set \cite{BLANCHINI19991747}]  \label{def:RCI}
A set $\Omega \subseteq X$ for an LTI system is an infinite-time viable set, also known as robust control invariant (RCI) set, if there exists a control policy $\mu$ such that $\Theta:= \mu(\Omega) \subseteq U$ and 
\begin{equation}
    \forall t\in \mathbb{N}, \forall x_t\in \Omega ,  \forall d_t \in D  \Rightarrow x_{t+1} \in \Omega.
\end{equation}
\end{definition}

\section{Problem Statement and Approach}

Consider a network of linear subsystems that are dynamically coupled in the following form:
\begin{equation}
\label{subsystems_LTV}
\begin{array}{ll}
     x_{i,t+1}= & A_{ii,t}x_{i,t} + B_{ii, t}u_{i,t} + d_{i,t}\\ & + \displaystyle \sum_{j\ne i}{A_{ij,t}x_{j,t}} + \displaystyle \sum_{j\ne i}{B_{ij,t}u_{j,t}},
    \end{array}
\end{equation}
where $x_{i,t} \in \mathbb{R}^{n_i}$, $u_{i,t} \in \mathbb{R}^{m_i}$, and $d_{i,t}\in D_{i,t}, D_{i,t} \subseteq \mathbb{R}^{n_i}$ are the state, control, and the disturbance for subsystem $i \in \mathcal{I}$, $\mathcal{I}= \mathbb{N}_{1:\eta}$, and $\eta$ is the number of subsystems. Also, $D_{i,t} = \mathcal{Z}(\bar{d}_{i,t}, G_{i,t}^{d})$, where $\bar{d}_{i,t} \in \mathbb{R}^{n_i}$ and $G_{i,t}^d \in \mathbb{R}^{n_i \times p_i^{d_t}}$ is the given disturbance set for subsystem $i$. 
We denote the disturbance set for the aggregate system as $D_t=\prod_{i\in \mathcal{I}} D_{i,t}$.
The matrices $A_{ii,t} \in \mathbb{R}^{n_i \times n_i}$ and $B_{ii,t} \in \mathbb{R}^{n_i \times m_i}$ characterize the (time-variant) internal dynamics of subsystem $i$. Also, $A_{ij,t} \in \mathbb{R}^{n_i \times n_j}$ and $B_{ij,t} \in \mathbb{R}^{n_i \times m_j}$ characterize the coupling effects of subsystem $j$ on subsystem $i$. 
\begin{remark}
If the sets $\mathcal{D}_{i,t}$ are not zonotopic, they can be over-approximated by zonotopic sets. Furthermore, if the disturbance set is given for the aggregated system, it may be decomposed into sets for each subsystem $i$, such that the aggregated disturbance set is the subset of the Cartesian product of the new sets for each subsystem.
\end{remark}
We consider coupled constraints over state $(X_t)$ and control input $(U_t)$. These sets can be polytopic sets or zonotopes. For the network of subsystems in \eqref{subsystems_LTV}, the constraints are
\begin{equation} \label{decomposition}
    [x_{1,t}^T , \cdots , x^T_{\eta,t}]^T \in X_t, \hspace{3mm}
    [u^T_{1,t} ,\cdots , u^T_{\eta,t}]^T \in U_t.
\end{equation}
We are now ready to formulate the problem of finding decentralized viable sets and their corresponding decentralized controllers $\mu_i(.)$ with respect to the constraints.

\begin{problem}\label{problem_viable}\textbf{(Decentralized Finite-time Viable Sets)}
Given a network of interconnected subsystems in the form of \eqref{subsystems_LTV} and the horizon $h \in \mathbb{N}_+$, find sets $\Omega_{i,t}$ and decentralized controllers $\mu_{i,t}(.), \forall i \in \mathcal{I}$ and $t \in \mathbb{N}_h$ such that $\prod_{i \in \mathcal{I}}\Omega_{i,t} \subseteq X_t$, $\prod_{i \in \mathcal{I}}\Theta_{i,t} \subseteq U_t$, and 
\begin{equation} \label{eq:subrci}
\begin{array}{ll}
        & \forall x_{i,t} \in \Omega_{i,t},  \forall u_{j,t}\in \Theta_{j,t}, (j\ne i),    \forall d_{i,t} \in D_{i,t} \\ & \Rightarrow x_{i,t+1}\in \Omega_{i,t+1},
    \end{array}
\end{equation}
where $\Theta_{i,t}=\mu_{i,t}(\Omega_{i,t})$.
\end{problem}
In essence, Problem \ref{problem_viable} is a tube-based trajectory optimization problem with reachability constraints for $h$ steps ahead, while taking into account all possible disturbances to ensure guaranteed reachability. This is similar to a robust MPC formulation with horizon $h$ when solved recursively, but without an objective function because the goal can be achieved by reachability set containment constraints. Thus, the objective function is optional. A terminal condition can also be introduced to the MPC formulation for the purpose of ensuring stability and recursive feasibility.

In addition to the previous problem, which focuses on the time-limited response of an LTV system, we also consider the infinite-time response of an LTI system:

\begin{problem} \label{Problem_RCI} \textbf{(Decentralized Infinite-time Viable Sets)}
For the particular case when each subsystem in (\ref{subsystems_LTV}) is time invariant, find sets $\Omega_i$ and decentralized controllers $\mu_i$(.) for $\forall i \in \mathcal{I}$, such that $\prod_{i \in \mathcal{I}}\Omega_i \subseteq X$, $\prod_{i \in \mathcal{I}}\Theta_i \subseteq U$, and
\begin{equation} \label{RCI_prob1}
\begin{array}{ll}
        & \forall x_{i,t} \in \Omega_i , \forall x_{j,t}\in \Omega_j , \forall u_{j,t}\in \Theta_j (j\ne i), \forall d_{i,t} \in D_i \\  & \Rightarrow x_{i,t+1}\in \Omega_i,
    \end{array}
\end{equation}
where $\Theta_i= \mu_i(\Omega_i)$. 
\end{problem}
It is worth noting that the concept of the infinite-time viable sets and Problem \ref{Problem_RCI} can be extended to $T$-periodic systems where $A_{t+T}=A_t, B_{t+T} = B_{t}, X_{t+T}=X_t, U_{t+T}=U_t, D_{t+T}=D_t, \forall t \in \mathbb{N}$. We omit studying this class of systems in this paper.

We wish to handle the coupled constraints \eqref{decomposition} separately. To this goal, we consider an additional step of decoupling the constraints. This allows each subsystem to impose constraints independently of the others, while still satisfying the coupled constraints when the subsystems are combined, therefore allowing for distributed computations. Consider the following domain and range for each controller:
\begin{equation}
\mu_{i,t}(.): X_{i,t} \rightarrow U_{i,t},    
\end{equation} 
where $X_{i,t}$ and $U_{i,t}$ are admissible sets in the state space and control space for subsystem $i$ at time $t$, respectively.
However, since the constraints are given in coupled forms, $X_{i,t}$ and $U_{i,t}$ are unknown and need to be found by decomposing $X_t$ and $U_t$, respectively. We address this in the following subproblem:
\begin{subproblem} [Set Decomposition] \label{subproblem}
Given a set $X \subset \mathbb{R}^{n}$ in the form of a zonotope and a set of integers $n_i$ ,$\forall i \in \mathcal{I}$ $(\sum_i n_i = n)$, find sets $X_i = \mathcal{Z}(c_i,G^x_i)$ where $c_i \in \mathbb{R}^{n_i}$ and $G^x_i \in \mathbb{R}^{n_i \times p^x_i} , p^x_i \in \mathbb{N}_+$ with the maximum volume for the set $\prod_{i \in \mathcal{I}}X_i$, such that 
$$\forall x_i\in X_i \Rightarrow [x_1^T, x_2^T, ..., x_\eta^T]^T \in X.$$
\end{subproblem}
{\bf Solution:} See Appendix.

Defining this subproblem for all time steps of $X_t$ and $U_t$ will lead to full decomposition of constraints.

\section{Assume-Guarantee Contracts}
\label{sec_AG}
In this section, we formalize assume-guarantee contracts for one system and provide details on the convex parameterization of contracts and controllers.  
\subsection{Definitions} 
\begin{definition}[Assume-Guarantee Contract] \label{definition_AG}
An assume-guarantee contract for system \eqref{singlesystem_LTV} is a pair $\mathcal{C} = (\mathcal{A},\mathcal{G})$, where:
\begin{itemize}
    \item $\mathcal{A}$ is the assumption, which is a sequence of disturbance sets $\mathcal{D}_t, t \in \mathbb{N}_{h-1}$;
    \item $\mathcal{G}$ is the guarantee, which is a sequence of tuples $(\mathcal{X}_t,\mathcal{U}_t)$, where $\mathcal{X}_t$ and $\mathcal{U}_t$ are two sets at time $t \in \mathbb{N}_h$ in the state space and control space, respectively.
\end{itemize}
\end{definition}

\begin{definition} [Contract Validity] \label{definition_validity}
A contract is \emph{valid} if its guarantee respects the system's constraints $\mathcal{X}_t \subseteq X_t, \forall t \in \mathbb{N}_h, \mathcal{U}_t \subseteq U_t, \forall t \in \mathbb{N}_{h-1}$.  
\end{definition}

\begin{definition} [Contract Satisfiability] \label{definition_satsifiablity}
A valid contract is \emph{satisfiable} if it is possible to find a control policy and viable sets such that $\Omega_t \subseteq \mathcal{X}_t, \forall t \in \mathbb{N}_h, \Theta_t \subseteq \mathcal{U}_t, \forall t \in \mathbb{N}_{h-1}$.
\end{definition}


\subsection{Finite Horizon Contract Satisfiability}
For a single system with a given disturbance bound, we show that a satisfiable contract can be found using convex programs, which encode a specific form of control policies. 
\begin{theorem} \label{Thrm_viable set_single} 
Given an LTV system in the form (\ref{singlesystem_LTV}) with the bounded disturbance set $D_t = \mathcal{Z}( \bar{d}_t , G^d_t )$, where $\bar{d}_t \in \mathbb{R}^n$ and $G^d_t \in \mathbb{R}^{n \times p_t}$, a finite horizon contract is satisfiable if $\exists k \in \mathbb{N}$, vectors $\bar{\mathrm{x}}_t\in \mathbb{R}^n , \bar{\mathrm{u}}_t \in \mathbb{R}^m$, and matrices $ T_t\in \mathbb{R}^{n \times l_t}$ and $M_t\in \mathbb{R}^{m \times l_t}$, where $l_0=k$ and $l_{t \ne 0} = k + \sum_{\hat{t}=0}^{t-1}{p_{\hat{t}}}$ such that the following relations hold:
\begin{subequations} \label{singleviablecon}
\begin{equation}\label{single_viable_constraint}
    [A_tT_t+B_tM_t , G^d_t] = T_{t+1}, \quad \forall t \in \mathbb{N}_{h-1}, 
\end{equation}
\begin{equation}\label{single_viable_constraint_center}
    A_t \bar{\mathrm{x}}_t+ B_t \bar{\mathrm{u}}_t + \bar{d}_t = \bar{\mathrm{x}}_{t+1}, \quad \forall t \in \mathbb{N}_{h-1},
\end{equation}
\begin{equation}\label{singleviable_statecon}
    \mathcal{Z}(\bar{\mathrm{x}}_t,T_t) \subseteq X_t, \quad \forall t \in \mathbb{N}_{h}, 
\end{equation}
\begin{equation}\label{singleviable_controlcon}
    \mathcal{Z}(\bar{\mathrm{u}}_t,M_t) \subseteq U_t, \quad \forall t \in \mathbb{N}_{h-1}. 
\end{equation}
\end{subequations}
Then $ \Omega_t=\mathcal{Z}(\bar{\mathrm{x}}_t,T_t) , t\in \mathbb{N}_h$ is the sequence of viable sets for horizon $h$ and $\Theta_t=\mathcal{Z}(\bar{\mathrm{u}}_t,M_t)$ is the sequence of action sets. Moreover, the controller $\mu_t(x_t)$ can be computed by:
\begin{subequations}\label{eq_control_finite}
\begin{equation}
    \mu_t(x_t)=\bar{\mathrm{u}}_t+M_t \zeta(x),
\end{equation}
\begin{equation}
    x_t=\bar{\mathrm{x}}_t+T_t\zeta(x),
\end{equation}
\begin{equation}
    \zeta \in \mathbb{B}_{l_t},
\end{equation}
\end{subequations}

\end{theorem}
\begin{proof}
The proof is by construction. Substituting \eqref{eq_control_finite} in (\ref{singlesystem_LTV}) yields:
\begin{equation}
    x_{t+1}     \in   A_t(T_t \zeta + \bar{\mathrm{x}}_t)+B_t(M_t \zeta + \bar{\mathrm{u}}_t ) + \bar{d}_t \oplus G^d_t \mathbb{B}_{p_t}, 
\end{equation}
where the right hand side set is equal to
\begin{equation} \label{next_set1}
    \{A_t \bar{\mathrm{x}}_t+ B_t \bar{\mathrm{u}}_t + \bar{d}_t\}  \oplus [A_t T_t + B_t M_t  ,G^d_t] \mathbb{B}_{l_t+p_t}.
\end{equation}
By definition, we know that
\begin{equation} \label{next_set2}
x_{t+1} \in \mathcal{Z}(\bar{\mathrm{x}}_{t+1} , T_{t+1}).
\end{equation}
By using the Minkowski sum property of zonotopes shown in \eqref{eq_zonotope_minkowski} and by assuming that the two sets shown in \eqref{next_set1} and \eqref{next_set2} are identical, it is straightforward to reach \eqref{single_viable_constraint} and \eqref{single_viable_constraint_center}. The next two constraints \eqref{singleviable_statecon} and \eqref{singleviable_controlcon} are imposing the constraints at each time step.
\end{proof}

Using the linear encoding for the set containment problem proposed in \cite{sadraddini2019linear}, containment constraints \eqref{singleviable_statecon} and \eqref{singleviable_controlcon} can be encoded as linear constraints with respect to $T_t, M_t, \bar{x}_t,$ and $ \bar{u}_t$. The cost function can be chosen depending on the application. We typically choose to minimize the summation of Frobenious norms of $T_t$ for $t\in \mathbb{N}_h$ as a heuristic to minimize the size of the viable sets. 

\begin{remark}
Note that the order of zonotope $\Omega_t$ is increasing at each time step. This makes the number of variables and constraints in the program grow quadratically with $h$. The complexity can be decreased by fixing the number of columns in matrices $T_t$ and $M_t$ at $k$ and by changing equation (\ref{single_viable_constraint}) to
\begin{equation}
    [A_t T_t + B_t M_t , G^d_t] =  [0_{n \times p_t} , T_{t+1}].
\end{equation}
However, this modification leads to a more conservative computation and may cause infeasibility.
\end{remark}

\subsection{Infinite Horizon Contract Satisfiability}
Inspired by the method in \cite{Rakovic2007}, we provide a linear programming approach to compute RCI sets. 
\begin{theorem} 
\label{thrm_single_RCI}
Given an LTI system in the form (\ref{singlesystem_LTV}) with the bouned disturbance set $D = \mathcal{Z}( \bar{d} , G^d )$, where $\bar{d} \in \mathbb{R}^n$ and $G^d \in \mathbb{R}^{n \times p}$, an infinite horizon contract is satisfiable if
 $\exists k \in \mathbb{N}, \beta \in [0,1)$, vectors $\bar{\mathrm{x}}\in \mathbb{R}^n , \bar{\mathrm{u}} \in \mathbb{R}^m$, and matrices $ T\in \mathbb{R}^{n \times k}$, $M\in \mathbb{R}^{m \times k}$, and $E\in \mathbb{R}^{n \times p}$, such that the following relations hold:
\begin{subequations}\label{singlercicon}
\begin{equation}\label{eq:thrm1}
    [AT+BM , G^d] = [ E , T],
\end{equation}
\begin{equation}\label{Econ}
    \mathcal{Z}(0,E) \subseteq \mathcal{Z}(0,\beta G^d), 
\end{equation}
\begin{equation}
    A\bar{\mathrm{x}} + B \bar{\mathrm{u}} + \bar{d} = \bar{\mathrm{x}},
\end{equation}
\begin{equation}\label{eq:thrm2} 
    \mathcal{Z}(\bar{\mathrm{x}},\dfrac{1}{1-\beta}T) \subseteq X,
\end{equation}
\begin{equation}\label{eq:thrm3}
    \mathcal{Z}(\bar{\mathrm{u}},\dfrac{1}{1-\beta} M) \subseteq U.
\end{equation}
\end{subequations}
Then $\Omega=\mathcal{Z}(\bar{\mathrm{x}},(1-\beta)^{-1}T)$ is a RCI set and the action set is $\Theta=\mathcal{Z}(\bar{\mathrm{u}},(1-\beta)^{-1}M)$. Furthermore, the controller can be computed by
\begin{subequations} \label{eq_control_infinite}
\begin{equation}
    \mu(x)=\bar{\mathrm{u}}+(1-\beta)^{-1} M \zeta(x),
\end{equation}
\begin{equation}
    x=\bar{\mathrm{x}}+(1-\beta)^{-1}T\zeta(x),
\end{equation}
\begin{equation}
    \zeta \in \mathbb{B}_k.
\end{equation}
\end{subequations}
\end{theorem}
\begin{proof}
Substituting policy \eqref{eq_control_infinite} in \eqref{singlesystem_LTV}, we obtain relation \eqref{eq:thrm1}. In order to prove invariance, we observe that:
\begin{equation}
    (AT+BM)\mathbb{B}_k \oplus \mathcal{Z}(0,G^d) \subseteq T\mathbb{B}_k \oplus \mathcal{Z}(0,\beta G^d), 
\end{equation}
and
\begin{equation}
    A\bar{\mathrm{x}} + B \bar{\mathrm{u}} + \bar{d} = \bar{\mathrm{x}}.
\end{equation}
We subtract $\mathcal{Z}(0,\beta G^d)$ from both sides in the Pontryagin difference sense \cite{kolmanovsky1998theory}. This is a valid operation when both sides are convex polytopes. We omit the proof as it is based on the properties of support functions \cite{rockafellar2015convex} of convex sets. We arrive at: 
\begin{equation}
\label{eq_middle_sadra}
    (AT+BM)\mathbb{B}_k \oplus (1-\beta) \mathcal{Z}(0,G^d) \subseteq T\mathbb{B}_k. 
\end{equation}
By multiplying both sides of \eqref{eq_middle_sadra} by $\dfrac{1}{1-\beta}$, we reach the conclusion that $ \mathcal{Z}(\bar{\mathrm{x}},(1-\beta)^{-1}T)$ is a RCI set with $\Theta=\mathcal{Z}(\bar{\mathrm{u}},(1-\beta)^{-1}M)$, and the proof is complete.
\end{proof}

Similar to \eqref{singleviablecon}, there exists a sufficient linear encoding for \eqref{singlercicon} to find the viable sets and action sets for fixed $k$ and $\beta$. The feasibility of the linear program implies the satisfiability of the contract.  

\begin{remark} \label{rmrk_RCI}
We can simplify Theorem \ref{thrm_single_RCI} by assuming $E = 0_{n \times p}$ and $\beta =0$. As a result, there is no need for constraint (\ref{Econ}). However, this assumption increases conservativeness and may lead to infeasibility.
\end{remark}

Note that, in both Theorem \ref{Thrm_viable set_single} and Theorem \ref{thrm_single_RCI}, the structures of matrices $T_t$ and $M_t$ (or $T$ and $M$) depend on a constant $k$ that acts like a hyper-parameter for our proposed linear program. Therefore, changing $k$ leads to a different result. This allows iterations over different $k$ to find a feasible solution for $T_t$ and $M_t$ (or $T$ and $M$). However, finding a feasible solution with the smallest possible value for $k$ is preferable, since it decreases the number of variables in the optimization problem.

\section{Composition of Parametric Assume-Guarantee Contracts}
\label{sec_AGR}
In this section, we focus on the network of coupled systems \eqref{subsystems_LTV} and provide the first steps for the solutions to Problems \ref{problem_viable} and \ref{Problem_RCI}. First, we decouple the subsystems by viewing the coupling effects of other subsystems as disturbances, which are called augmented disturbances $d_{i,t}^{aug}$ and are equal to:
\begin{equation} \label{disturbance_coupling_effect}
    d_{i,t}^{aug} := \sum_{j \ne i} A_{ij,t}x_{j,t} + \sum_{j \ne i} B_{ij,t}u_{j,t}+ d_{i,t},
\end{equation}
where $i$ is the subsystem's index and $t$ is the time step. Next, for each subsystem $i\in \mathcal{I}$, we define one AG contract denoted by $\mathcal{C}_i$. The collection of these contracts forms a set of AG contracts $\mathcal{C} = \{ \mathcal{C}_i|i\in \mathcal{I}\}$ for a network of coupled subsystems. From \eqref{disturbance_coupling_effect}, it can be seen that, for subsystem $i$, the assumption set over the disturbance space, which is denoted by $W_{i,t}$, can be defined as a function of the guarantees of other subsystems $(\mathcal{X}_{j,t},\mathcal{U}_{j,t}) (j \neq i)$ as follows:
\begin{equation}\label{eq_d_aug}
    W_{i,t}  := \bigoplus_{j \ne i} A_{ij,t}\mathcal{X}_{j,t} \oplus \bigoplus_{j \ne i} B_{ij,t}\mathcal{U}_{j,t} \oplus D_{i,t}.
\end{equation}
However, unlike the single-system case, in which by using Theorem \ref{Thrm_viable set_single} (or Theorem \ref{thrm_single_RCI} for the LTI case), we could directly compute the guarantee for a given disturbance set, in the case of interconnected subsystems, finding a set of satisfiable contracts is not straightforward. This is due to the fact that the guarantee of one subsystem affects the assumptions of other subsystems as a result of treating the coupling effects as disturbance. This is known as the circularity problem of AG contracts.

Also, it is worth noting that the contracts are common knowledge among all subsystems because each subsystem has to be able to compute its assumed disturbance using \eqref{eq_d_aug}. This, in fact, helps in the compositional computation of the synthesis problem.

Next, we introduce a correctness criterion to break the circularity of AG contracts and parametric contracts to search over contract sets to find a set of satisfiable contracts. Also, we propose a potential function as a quantitative indicator of how far a set of given contracts is from correct composition. Finally, we introduce a set of sufficient constraints for the validity of parametric contracts.

\subsection{Composition Correctness}
From \eqref{disturbance_coupling_effect}, it can be shown that after finding the viable sets $\Omega_{i,t}$ and action sets $ \Theta_{i,t}$ for all $i\in \mathcal{I}$, the actual augmented disturbance set for subsystem $i$ at time $t$, denoted by $D_{i,t}^{aug}$, can be obtained as:
\begin{equation}\label{actual_disturbance}
    D^{aug}_{i,t} := \bigoplus_{j \ne i} A_{ij,t}\Omega_{j,t} \oplus \bigoplus_{j \ne i} B_{ij,t}\Theta_{j,t} \oplus D_{i,t}.
\end{equation}
If the actual experienced disturbance set for every subsystem becomes a subset of its assumed disturbance set, the circularity of AG contracts is broken. We say that the {\em composition of the contracts is correct}. The formal definition of correctness is given as follows:
\begin{definition}[Composition Correctness]
Consider a set of locally valid assume-guarantee contracts $\mathcal{C}_i = (\mathcal{A}_i,\mathcal{G}_i)$, $i \in \mathcal{I}$. The composition is correct if the following relation holds:
\begin{equation}\label{composition_correctness_criterion}
D_{i,t}^{aug} \subseteq W_{i,t}   , \forall i \in \mathcal{I} , \forall t \in \mathbb{N}_{h-1} 
\end{equation}
\end{definition}
However, in order to make compositional computations easier, relation \eqref{composition_correctness_criterion} is replaced with the following more conservative version:
\begin{equation}\label{composition_correctness_criterion_step2}
    \Omega_{i,t} \subseteq \mathcal{X}_{i,t} , \Theta_{i,t} \subseteq \mathcal{U}_{i,t} , \forall i \in \mathcal{I} , \forall t \in \mathbb{N}_{h-1},
\end{equation}
When \eqref{composition_correctness_criterion_step2} holds, from \eqref{eq_d_aug} and \eqref{actual_disturbance}, it is evident that \eqref{composition_correctness_criterion} holds as well.
Our characterization of correctness is Boolean, i.e., it indicates whether or not the contract composition is correct. Nevertheless, having a quantitative measure that describes how far a set of given contracts is from correct composition is desirable. Such measurement can be very helpful in directing us toward correct contract composition. The following ``potential function'' does exactly that by assigning a score to a set of contracts.

\begin{definition}[Potential Function]
Given a set of contracts $\mathcal{C} = \{\mathcal{C}_i |  i \in \mathcal{I}\}$, its potential function is
\begin{equation}
    \mathcal{V}(\mathcal{C}) = \sum_{i \in \mathcal{I}} \mathcal{V}_i(\mathcal{C}),
\end{equation}
where $\mathcal{V}_i(\mathcal{C})$ is defined as follows:
\begin{equation} \label{directed_iasdorf_distance_x_u_finite_time}
    \mathcal{V}_i(\mathcal{C}) := \sum_{t\in \mathbb{N}_{h}} d_{DH} ( \mathcal{X}_{i,t} , \Omega_{i,t} ) +  \sum_{t\in \mathbb{N}_{h-1}} d_{DH} ( \mathcal{U}_{i,t} , \Theta_{i,t} ).
\end{equation}
\end{definition}

The potential function is equal to the sum of the directed Hausdorff distances between the sets in \eqref{composition_correctness_criterion_step2}. Given that \eqref{composition_correctness_criterion_step2} leads to \eqref{composition_correctness_criterion} and the potential function is zero if and only if \eqref{composition_correctness_criterion_step2} holds, it can be inferred that when the potential function is zero, the composition of the contracts is correct, implying that each system anticipates a larger set of disturbance than it will actually experience.  Thus, this quantitative characterization of correct composition is sound. Fig \ref{potential_function_x_fig} shows an example of the sets in the state space of subsystem $i$ at a given time $t$. A similar diagram can be drawn for the control space to show the other component of the potential function.

\begin{figure} 
\begin{center}
\includegraphics[scale=0.7]{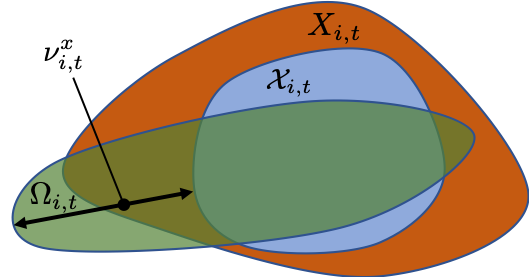}    
\caption{The admissible set $X_{i,t}$ (red), guarantee set in the state space $\mathcal{X}_{i,t}$ (blue), and the viable set $\Omega_{i,t}$ (green) for a 2-dimensional example subsystem $i$ at time $t$. The first component of the potential function in \eqref{directed_iasdorf_distance_x_u_finite_time} at time step $t$ (i.e. $d_{DH} ( \mathcal{X}_{i,t} , \Omega_{i,t} )$) is denoted by $\nu_{i,t}^x$. If $\nu_{i,t}^x=0$, then $\Omega_{i,t} \subseteq \mathcal{X}_{i,t}$. Our goal is to find a set of contracts such that $ \Omega_{i,t} \subseteq \mathcal{X}_{i,t} \subseteq X_{i,t}$ for all $i,t$.}  
\label{potential_function_x_fig}                                 
\end{center}                                 
\end{figure}

\subsection{Parametric Contracts}
To be able to search over the contract sets, we introduce parametric sets. The baseline sets are zonotopes $\mathcal{Z}(\bar{c}^x_{i,t} , C^x_{i,t})$ and $  \mathcal{Z}(\bar{c}^u_{i,t} , C^u_{i,t})$,
for all $i \in \mathcal{I}$ and $t \in \mathbb{N}_h$, where $\bar{c}^x_{i,t} \in \mathbb{R}^{n_i}$, $\bar{c}^u_{i,t} \in \mathbb{R}^{m_i}$, and $C^x_{i,t} \in \mathbb{R}^{n_i \times \zeta^x_{i,t}}$, $C^u_{i,t} \in \mathbb{R}^{m_i \times \zeta^u_{i,t}}$. The baseline sets need to be defined by the user in advance. One option is to select them as the viable sets and action sets for each subsystem, respectively, while ignoring the couplings to other subsystems. The parametric sets are denoted by $\mathcal{X}_{i,t}(\alpha^x_{i,t}) $ and $\mathcal{U}_{i,t}(\alpha^u_{i,t}) $, and are defined as follows:
\begin{subequations} \label{parameterized_sets}
\begin{equation}
    \mathcal{X}_{i,t}(\alpha^x_{i,t}) := \mathcal{Z}(\bar{c}^x_{i,t} , C^x_{i,t} \diag(\alpha^x_{i,t})),
\end{equation}
\begin{equation}
    \mathcal{U}_{i,t}(\alpha^u_{i,t}) := \mathcal{Z}(\bar{c}^u_{i,t} , C^u_{i,t} \diag(\alpha^u_{i,t})),
\end{equation}
\end{subequations}
where parameters $\alpha^x_{i,t} \in \mathbb{R}_{+}^{\zeta^x_{i,t}}$ and $\alpha^u_{i,t} \in \mathbb{R}_{+}^{\zeta^u_{i,t}}$ are vectors with non-negative real entries that scale the generators' columns of the baseline sets. Also, the set of all parameters is denoted by $\alpha=\{\alpha^x_{i,t},\alpha^u_{i,t}\}_{\forall i , t}$. Finally, a parametric contract $\mathcal{C}_i$ has its guarantee in the form of \eqref{parameterized_sets} in the state and control space, respectively, and its assumption can be directly derived by replacing \eqref{parameterized_sets} in \eqref{eq_d_aug}, which is shown by $W_{i,t}(\alpha)$.
\subsection{Parametric Potential Function} \label{parametric_potential_function_section}
Here we introduce parametric potential functions, which quantify 
how close a set of contracts/parameters is from composition correctness when the contracts are given in parametric form.
\begin{definition}[Parametric Potential Function]
Given a set of parameters $\alpha$, the parametric potential function is defined as:
\begin{equation}
\label{eq_v_alpha}
    \mathcal{V}(\alpha) = \sum_{i \in \mathcal{I}} \mathcal{V}_i(\alpha),
\end{equation}
where
\begin{equation}\label{eq_v_i_alpha}
\begin{array}{ll}
        \mathcal{V}_i(\alpha) := &  \displaystyle \sum_{t\in \mathbb{N}_{h}}  d_{DH}(\mathcal{X}_{i,t}(\alpha^x_{i,t}) ,  \Omega_{i,t} )+ \\ &   \displaystyle \sum_{t\in \mathbb{N}_{h-1}}  d_{DH}(\mathcal{U}_{i,t}(\alpha^u_{i,t}) ,  \Theta_{i,t} ).
\end{array}
\end{equation}
\end{definition}
Note that the parametric potential function is zero if and only if
\begin{equation} \label{corectness_prime}
    \Omega_{i,t} \subseteq \mathcal{X}_{i,t}(\alpha^x_{i,t}) \quad\text{and}\quad \Theta_{i,t} \subseteq \mathcal{U}_{i,t}(\alpha^u_{i,t}).
\end{equation}
Thus the level set of the potential function at zero equals the set of parameters that correspond to the correct contract composition. These parameters are referred to as \textit{correct parameters}.
Using the linear program \eqref{directed_hausdorff_distance_computation}, the optimization problem for computing $\mathcal{V}_i(\alpha)$ is as follows: 
\begin{subequations} \label{LTV_guarantees}
\begin{align}
    & \mathcal{V}_i(\alpha) = \min_{\bar{\mathrm{x}}^i,T^i,\bar{\mathrm{u}}^i,M^i , d_t^x , d_t^u}
    \begin{aligned}[t]
       & \sum_{t \in \mathbb{N}_{h}} d^x_t + \sum_{t \in \mathbb{N}_{h-1}} d^u_t 
    \end{aligned} \notag \\
    &\text{subject to} \notag \\
    & [A_{ii,t}T^i_t+ B_{ii,t} M^i_t , \mathcal{W}_{i,t}^{aug} ] =  T^i_{t+1}, \quad \forall t \in \mathbb{N}_{h-1} \label{eq:a} \\
    & A_{ii,t}\bar{\mathrm{x}}^i_t + B_{ii,t} \bar{\mathrm{u}}^i_t + \bar{d}_{i,t}^{aug} = \bar{\mathrm{x}}^i_{t+1},  \quad  \forall t \in \mathbb{N}_{h-1} \label{eq:b}\\
    & \mathcal{Z}(\bar{\mathrm{x}}^i_t,T^i_t) \subseteq \ \mathcal{X}_{i,t}(\alpha^x_{i,t}) \oplus \mathcal{Z}(0_{n_i},d^x_t I_{n_i}) ,   \forall t \in \mathbb{N}_{h} \label{eq:d}\\
    & \mathcal{Z}(\bar{\mathrm{u}}^i_t,M^i_t) \subseteq \ \mathcal{U}_i(t,\alpha^u_{i,t}) \oplus \mathcal{Z}(0_{m_i},d^u_t I_{m_i}) , \quad \forall t \in \mathbb{N}_{h-1} \label{eq:e}\\
    & d_t^x\geq 0 , \hspace{2 mm} \forall t \in \mathbb{N}_{h}, \label{eq:hausdorffdistance_x}\\
    & d_t^u \geq 0 , \hspace{2 mm} \forall t \in \mathbb{N}_{h-1}, \label{eq:hausdorffdistance_u}
\end{align}
\end{subequations}
where $\bar{\mathrm{x}}^i,T^i,\bar{\mathrm{u}}^i$, and $M^i$ are sets containing all $\bar{\mathrm{x}}^i_t,T^i_t,\bar{\mathrm{u}}^i_t$, and $M^i_t$, respectively. Constraints (\ref{eq:a}) and (\ref{eq:b}) come from Theorem \ref{Thrm_viable set_single}, which enforce viability conditions (must be replaced by the constraints in Theorem \ref{thrm_single_RCI} in the LTI case), where $\bar{d}_{i,t}^{aug}$ and $\mathcal{W}_{i,t}^{aug}$ are the center and generator of the assumed disturbance set for the subsystem $i$, respectively:
\begin{equation}\label{reduced_dist_set}
\mathcal{Z}( \bar{d}_{i,t}^{aug} , \mathcal{W}_{i,t}^{aug})  := \reduce( W_{i,t}(\alpha) ),\forall t \in \mathbb{N}_{h-1}.
\end{equation} 
Since the assumed disturbance set for each subsystem is known for a given $\alpha$, zonotope order reduction methods may be used to over-approximate the disturbance set in order to reduce computational complexity. Also, scalar variables $d^x_t$ and $d^u_t$, constraints (\ref{eq:d})-(\ref{eq:hausdorffdistance_u}) and the objective function originate from Lemma \ref{hasdurff_distance_computation} in order to compute the sum of Directed Hausdorff distances over all time steps on both state and control spaces. 
The original problem of computing all of the directed Hausdorff distances is a multi-objective optimization problem, i.e. one objective function for each Directed Hausdorff distance computation. However, since the optimal points of the optimization problems are not conflicting with each other, we may add the objective functions and ensure that the outcome is the same as when each optimization problem is solved independently. Note that constraints \eqref{eq:d} and \eqref{eq:e} can be encoded into linear constraints using Lemma \ref{lemma_weighted}.
The following theorem is the main result of this section:

\begin{theorem} \textbf{(Convexity of parametric potential function)} Using our parameterization \eqref{parameterized_sets} and the linear encoding with containment introduced in Lemma \ref{sadra_zon_containment} and Lemma \ref{lemma_weighted}, the parametric potential function $\mathcal{V}(\alpha)$ is a convex piecewise affine function. The set of correct parameters, which is the level set of $\mathcal{V}(\alpha)$ at zero, is also a convex set.
\end{theorem}

\begin{proof}
As shown in (\ref{LTV_guarantees}), each $\mathcal{V}_i(\alpha)$ is formulated in a linear program, implying that each $\mathcal{V}_i(\alpha)$ is a convex piecewise affine function \cite{bertsimas1997introduction}. Also, because the summation of convex piecewise affine functions remains convex and piecewise affine, $\mathcal{V}(\alpha)$ is also a convex piecewise affine function. Since the level set of a convex function is a convex set, the set of correct parameters is a convex set.
\end{proof}

\subsection{Validity}\label{composition_validity}

A set of desirable contracts needs to be \emph{compositionally correct} and individually \emph{satisfiable} and \emph{valid}, as indicated in Section \ref{sec_AG}. The set of parameters that results in composition correctness and satisfiability does not necessarily lead to validity as well. In other words, in Fig. \ref{potential_function_x_fig}, it may be possible that $\Omega_{i,t} \subseteq \mathcal{X}_{i,t}$ but $\Omega_{i,t} \not\subseteq X_{i,t}$. In this subsection, the focus is on finding the set of valid parameters, which are the set of parameters that corresponds to a valid contract for each subsystem. A contract is valid if and only if:
\begin{subequations}\label{valid_parameter_set}
\begin{equation}
    \mathcal{X}_{i,t}(\alpha^x_{i,t}) \subseteq X_{i,t}, \forall t \in \mathbb{N}_h,
\end{equation}
\begin{equation}
    \mathcal{U}_{i,t}(\alpha^u_{i,t}) \subseteq U_{i,t},\forall t \in \mathbb{N}_{h-1}.
\end{equation}
\end{subequations}
\begin{theorem}[Convexity of valid parameters set]
Using parameterization \eqref{parameterized_sets} and the linear encoding with containment introduced in Lemma \ref{lemma_weighted}, the set of valid parameters derived from \eqref{valid_parameter_set} is a convex set.
\end{theorem}
\begin{proof}
The set of valid parameters is the feasible region of the linear encoding of \eqref{valid_parameter_set}. Therefore, it is a convex set \cite{bertsimas1997introduction}.
\end{proof}
Notably, in a centralized method, we may avoid decoupling constraints via the solution for Subproblem \ref{subproblem} and instead impose them by $\prod_{i \in \mathcal{I}}\mathcal{X}_{i,t}(\alpha^x_{i,t}) \subseteq X_t, \forall t \in \mathbb{N}_h$ and $\prod_{i \in \mathcal{I}}\mathcal{U}_{i,t}(\alpha^u_{i,t}) \subseteq U_t,\forall t \in \mathbb{N}_{h-1}$.

\section{Compositional Synthesis and Computations}
\label{sec_gradient}
Two approaches based on parametric assume-guarantee contracts are proposed in this section to solve Problems \ref{problem_viable} and \ref{Problem_RCI}. The first provides a single centralized optimization to identify a set of desirable contracts, decentralized viable sets, and decentralized controllers. The second method uses the notion of parametric potential function to achieve the same task in a compositional fashion.
\subsection{Single Convex Program} \label{centralized_approach}

By combining all of the encodings presented so far, we obtain the following centralized linear program:
\begin{subequations} \label{LTV_optmz}
\begin{align}
    & \Omega , \Theta = \underset{\bar{\mathrm{x}}^i_t,T^i_t,\bar{\mathrm{u}}^i_t,M^i_t,\alpha , \mathrm{d}_i , \mathrm{G}_i}{\text{argmin}}
    \begin{aligned}[t]
       & \sum_{t \in \mathbb{N}_{h-1}}\sum_{i\in \mathcal{I}}{\text{sum}(\alpha_{i,t}^x)}
    \end{aligned} \notag \\
    &\text{subject to} \notag \\
    & [A_{ii,t}T^i_t+ B_{ii,t} M^i_t , \mathrm{G}_{i,t} ] =  T^i_{t+1} ,    \forall t\in \mathbb{N}_{h-1} , \forall i \in \mathcal{I}  \label{eq2:a} \\
    & A_{ii,t}\bar{\mathrm{x}}^i_t + B_{ii,t} \bar{\mathrm{u}}^i_t + \mathrm{d}_{i,t} = \bar{\mathrm{x}}^i_{t+1},  \forall t\in \mathbb{N}_{h-1} ,  \forall i \in \mathcal{I} \label{eq2:b}\\
    & \mathcal{Z}( \mathrm{d}_{i,t} , \mathrm{G}_{i,t})  =  W_{i,t}(\alpha),  \forall t \in \mathbb{N}_{h-1} , \forall i \in \mathcal{I} \label{eq2:e} \\
    & \prod_{i\in \mathcal{I}} \mathcal{Z}( \bar{\mathrm{x}}^i_t , T_t^i ) \subseteq X_t , \quad \forall t\in \mathbb{N}_{h}        \label{eq2:c}\\
    & \prod_{i\in \mathcal{I}} \mathcal{Z}( \bar{\mathrm{u}}^i_t , M_t^i ) \subseteq U_t, \quad \forall t\in \mathbb{N}_{h-1}   \label{eq2:d}\\
    & \mathcal{Z}(\bar{\mathrm{x}}_t^i,T^i_t) \subseteq \mathcal{X}_{i,t}(\alpha^x_{i,t}), \quad \forall t\in \mathbb{N}_{h-1} , \forall i \in \mathcal{I}     \label{eq2:f}\\
    & \mathcal{Z}(\bar{\mathrm{u}}^i_t,M^i_t) \subseteq \mathcal{U}_{i,t}(\alpha^u_{i,t}), \quad \forall t\in \mathbb{N}_{h-1} , \forall i \in \mathcal{I}   \label{eq2:g}\\
    & \alpha^x_{i,t},\alpha^u_{i,t} \geq 0,  \hspace{4 mm} \forall t\in \mathbb{N}_{h-1} ,  \forall i \in \mathcal{I} \label{eq2:h},
\end{align}
\end{subequations}
where $\Omega = \{\Omega_i| \Omega_i= \mathcal{Z}(\bar{\mathrm{x}}^i_0,T^i_0) , ... ,  \mathcal{Z}(\bar{\mathrm{x}}^i_h,T^i_h) ,  \forall i \in \mathcal{I} \} $ and $\Theta = \{\Theta_i | \Theta_i = \mathcal{Z}(\bar{\mathrm{u}}^i_0,M^i_0) , ... , \mathcal{Z}(\bar{\mathrm{u}}^i_{h-1},M^i_{h-1}) , \forall i \in \mathcal{I}  \} $ are decentralized viable sets and action sets, respectively. The constraints \eqref{eq2:a} and \eqref{eq2:b} are sufficient constraints for viable sets introduced in Theorem \ref{Thrm_viable set_single} (must be replaced with the constraints in Theorem \ref{thrm_single_RCI} for Problem \ref{Problem_RCI}). The constraints on the state and control input are imposed in \eqref{eq2:c} and \eqref{eq2:d}. They are applied over the aggregated system in order to reduce the conservatism caused by decoupling the coupled constraints. 
The constraint \eqref{eq2:e} calculates the assumed disturbance set from \eqref{eq_d_aug}. We cannot use zonotope order reduction techniques on this formulation since doing so will result in a non-convex encoding. The constraints \eqref{eq2:f} and \eqref{eq2:g} are also the proposed requirements suggested by the correctness criterion \eqref{corectness_prime}. Lemma \ref{lemma_weighted} may be used to encode these two constraints into linear constraints. It is worth noting that there is no need to provide validity constraints for parametric AG contracts because constraints \eqref{eq2:c} and \eqref{eq2:d} suffice.
Furthermore, the objective function is ad-hoc and it may be selected as shown in \eqref{LTV_optmz}, which is a heuristic strategy for decreasing the volume of the viable sets. The constant $k$ (introduced in Theorem \ref{Thrm_viable set_single}) must be established before solving the linear program in \eqref{LTV_optmz}. Our strategy is to start with a small initial $k$ and solve \eqref{LTV_optmz}, then increase it by one unit until feasibility is attained.

Note that this method is sound, because the correctness criterion is enforced in the process by using zonotope containment constraints. As a result, the output is correct-by-construction. Also, as $k$ goes to infinity, our approach is theoretically able to find a set of correct parameters, if one exists.

\subsection{Compositional Approach} \label{compos_method}

The centralized method proposed in the previous section suffers from the curse of dimensionality. Despite the fact that the control is decentralized, solving a single, very large linear program is impractical for large-scale systems. The computational complexity is driven by two main factors: (i) the large number of variables and constraints in the single optimization problem \eqref{LTV_optmz}; (ii) the order of the zonotope $\mathcal{Z}( \bar{d}_{i,t}^{aug} , D_{i,t}^{aug})$ in the constraint \eqref{eq2:e}, which becomes very large when the number of neighboring subsystems is large due to the Minkowski sum computations in \eqref{eq_d_aug}, which makes \eqref{eq2:a} infeasible for thin $T_t^i$ matrices. Also, because of the existing recursivity in \eqref{eq2:a}, the widths of the matrices $T_t^i$ increase with time steps. 

The main contribution of this paper is a compositional method for the computation of the viable sets. In this subsection, we show  that, by using our proposed parameterized sets \eqref{parameterized_sets} and the convex potential function \eqref{eq_v_alpha}, the single optimization problem in the previous subsection can be transformed into a number of iterative small linear optimization problems with convergence guarantees. Also, following \eqref{LTV_guarantees}, we use zonotope order reduction methods to lower the order of the augmented disturbance sets, which substantially speeds up the calculation.

As explained before, when the potential function is zero, the parameters are inside the correct set of parameters and we have a set of compositionally correct cont,racts. Thus, the goal is to determine $\alpha^*$ such that:
\begin{equation} \label{potential_optmz}
\mathcal{V}^* = \underset{\alpha}{\text{min   }} \sum_{i \in \mathcal{I}}\mathcal{V}_i(\alpha) ,
\end{equation}
where $\mathcal{V}_i(\alpha)$ is the potential function that corresponds to subsystem $i$. Each subsystem concurrently computes its own portion of the potential function by \eqref{LTV_guarantees}. When the optimal value of the above optimization problem is greater than zero $(\mathcal{V}^* > 0)$, the optimal parameters are not in the set of correct parameters. In this case, the algorithm raises the hyper-parameter $k$ (defined in Theorem \ref{Thrm_viable set_single}) by one unit and repeats the same procedure, until it either hits a user-defined maximum threshold for $k$ (no solution found) or it finds parameters with zero value of the potential function. We solve the aforementioned optimization problem using gradient descent. We start with a random initial guess for $\alpha$, and then update the parameters using:
\begin{equation} \label{gradient_descent}
    \alpha \leftarrow \alpha - \delta  \sum_{i \in \mathcal{I}}\nabla_{\alpha} \mathcal{V}_i(\alpha),
\end{equation}
where $\delta$ is the step size. We use duality to calculate the gradient of the potential function. The optimal dual variable of a constraint in a linear program is equal to the derivative of the objective function with respect to the right hand side of the corresponding constraint, according to \textit{sensitivity analysis} of linear programs \cite{bertsimas1997introduction}. Using this property and the chain rule, we can compute $\nabla_{\alpha} \mathcal{V}_i(\alpha)$ after solving the linear program proposed for $\mathcal{V}_i(\alpha)$. In our implementation, we used the built-in function in Gurobi \cite{gurobi} to find the optimal dual variables of the constraints that contain the parameters. It should be noted that in \eqref{reduced_dist_set}, additional tighter zonotope order reduction approaches can also be used to further minimize conservatism. In our implementation, we used the Boxing method to keep the chain rule simpler. 

There are $\eta$ gradients for each point in the parameter space, where $\eta$ is the number of subsystems. In other words, each subsystem selfishly proposes a direction that best suits it at each point in the parameters space. By calculating the sum of these directions, we can determine which direction is the best for the aggregated system. This is because the suggested objective function is in the form of $\sum_{i \in \mathcal{I}} \mathcal{V}_i(\alpha)$, and its derivative with respect to $\alpha$ is equal to the sum of the gradients computed by each subsystem. Finally, we stop updating the parameters when the value of the potential function reaches a plateau and there is no further improvement in the objective function after the set of parameters is updated.

As stated in Subsection \ref{composition_validity}, the existence of a zero potential function does not imply that the set of parameters is also valid. There are two ways to address this issue. First, we can add the directed Hausdorff distance between parameterized sets and admissible sets to the objective function in \eqref{LTV_guarantees}, as we did for correctness. Then, we can determine a set of parameters that will cause the new objective function to reach zero. Second, we can limiting the domain of the potential function to the set of valid parameters. To avoid leaving the valid set of parameters after each update, the parameters must be projected to the set of valid parameters. In this case, the potential function remains convex because the set of valid parameters is convex. The second method projects the parameters onto the target set in a single step, whereas the first requires many iterations. Here, we pick the second option due to its speed. The projection is performed via the following optimization problem:
\begin{equation}\label{alpha_projection}
\begin{aligned}
    \alpha^{x^p}_{i,t} = \underset{\alpha^p}{\text{argmin }} \quad & ||\alpha^{x^p}_{i,t} -\alpha^x_{i,t}||_2 \\
    \textrm{s.t.} \quad & \mathcal{X}_{i,t}(\alpha^{x^p}_{i,t}) \subseteq X_{i,t},
\end{aligned}
\end{equation} 
which projects $\alpha_{i,t}^x$ to $\alpha^{x^p}_{i,t}$ for parameters corresponding to the state space of subsystem $i$ at time step $t$. There is a similar QP for the parameters corresponding to the control space for each subsystem as well. The algorithm is shown in Fig. \ref{fig:flowchart} as a flowchart.








  
  
  
  
  

  
    



\tikzset{%
  >={Latex[width=2mm,length=2mm]},
            base/.style = {rectangle, rounded corners, draw=black,
                          minimum width=4cm, minimum height=1cm,
                          text centered, font=\sffamily},
  activityStarts/.style = {base, fill=blue!30},
      startstop/.style = {base, fill=green!30 , minimum width=1.5cm, minimum height=1cm},
    activityRuns/.style = {base, fill=green!30},
         process/.style = {base, minimum width=1.5cm, fill=orange!15,
                          font=\ttfamily},
        sum/.style = {circle , fill=yellow!20, draw=black}
}
  
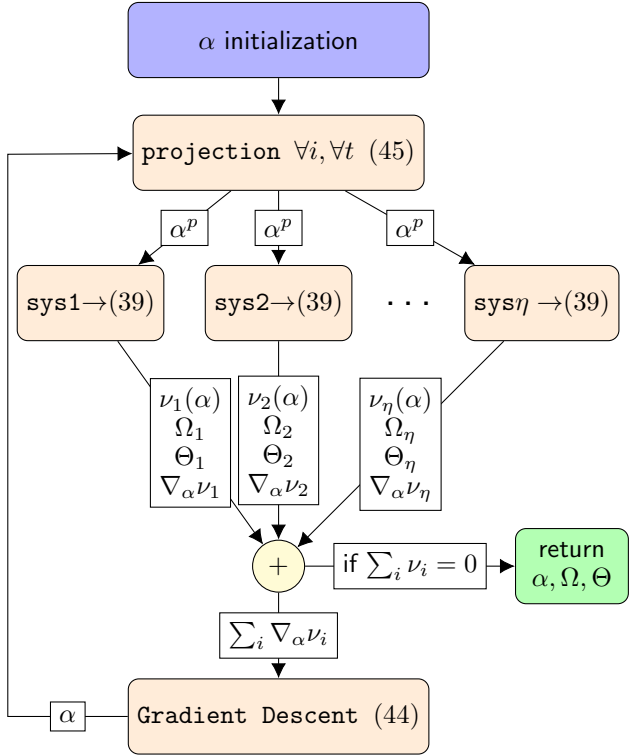
\begin{figure}
\centering
\begin{tikzpicture}[node distance=1.5cm,
    every node/.style={fill=white, font=\sffamily}, align=center]
  \node (alpha_init)             [activityStarts]              {$\alpha$ initialization};
  \node (projection)     [process, below of=alpha_init]          { projection $\forall i,\forall t$ \eqref{alpha_projection}};
  \node (nu_1)      [process , below of=projection , xshift=-2.5cm , yshift=-0.5cm]    {sys1$\rightarrow$\eqref{LTV_guarantees}};
  \node (nu_2)      [process, right of=nu_1 , xshift=1.0cm ]   {sys2$\rightarrow$\eqref{LTV_guarantees}};
  \node (nu_eta)      [process, right of=nu_2 , xshift= 2cm]   {sys$\eta\rightarrow$\eqref{LTV_guarantees}};
  \node (sum)       [sum, below of= projection , yshift= -4cm ]      {$+$};
  \node (GD)        [process , below of=sum , yshift = -0.5cm]        {Gradient Descent \eqref{gradient_descent}};
  \node (terminate)     [startstop , right of=sum , xshift= 2.4cm]        {return \\ $\alpha,\Omega,\Theta$};
  \draw[->]             (alpha_init) -- (projection);
  \draw[->]     (projection) -- node [rectangle, draw=black]  {$\alpha^p $} (nu_1);
  \draw[->]     (projection) -- node [rectangle, draw=black] {$\alpha^p $} (nu_2);
  \draw[->]     (projection) -- node [rectangle, draw=black] {$\alpha^p $} (nu_eta);
  \path (nu_2) -- node[auto=false]{\textbf{. . .}} (nu_eta);
  \draw[->]      (nu_1) -- node [rectangle , draw=black] { $\nu_1(\alpha)$ \\ $\Omega_1$ \\ $\Theta_1$ \\ $\nabla_\alpha\nu_1$} (sum) ;
  \draw[->]      (nu_2) -- node [rectangle , draw=black] { $\nu_2(\alpha)$ \\ $\Omega_2$ \\ $\Theta_2$ \\ $\nabla_\alpha\nu_2$} (sum);
  \draw[->]      (nu_eta) -- node [rectangle , draw=black] { $\nu_\eta(\alpha)$ \\ $\Omega_\eta$ \\ $\Theta_\eta$ \\ $\nabla_\alpha\nu_\eta$ } (sum);
  \draw[->]     (sum) -- node [rectangle, draw=black] {$\sum_i \nabla_\alpha\nu_i$} (GD);
  \draw[->]     (GD)  -- node [rectangle, draw=black] {$\alpha$} ++(-3.6,0) -- ++(0,3.74) -- ++(0,3.74) -- (projection);
  \draw[->]     (sum) -- node [rectangle, draw=black] {if $\sum_i \nu_i =0$}(terminate);

  \end{tikzpicture}
\caption{Flowchart representation of the algorithm from Sec. \ref{compos_method}. The parameters are initialized and then projected to the set of valid parameters. Then each subsystem independently determines the AG contracts and its compartment of the potential function. One of the following three scenarios occurs once the potential function has been determined: (i) If it is zero, the result is returned; (ii) If it hits a plateau, $k$ is increase by one unit and the algorithm is restarted with the most recently updated parameters (this step is not shown); (iii) If (i) or (ii) do not apply, the gradient information is used to update the parameters and the the algorithm reiterates.} \label{fig:flowchart}
\end{figure}

\section{Distributed Robust MPC With Stability Guarantee} 
\label{section:mpc}

In this section, we show that our solution for Problem \ref{problem_viable} can be used for the problem of Distributed Robust MPC when solved recursively. Each MPC needs to have: (i) a given initial state $x_{i,0}, \forall i \in \mathcal{I}$, (ii) a quadratic cost function $J_i := \sum_t {\mathrm{x}^i_t}^T Q \mathrm{x}^i_t + {\mathrm{u}^i_t}^T R \mathrm{u}^i_t$ for each subsystem $i$, and (ii) the RCI sets $(\Omega_i)$ as terminal sets. In \cite{chisci1996dual,scokaert1999suboptimal}, it is shown that considering a terminal condition for an MPC ensures the stability and recursive feasibility of the approach. The terminal condition specifies a target set for the MPC's last time step, which is encoded using set containment constraints similar to those described in the preceding sections. Using the provided solution for problem \ref{Problem_RCI}, the RCI sets can be determined in offline mode. This section incorporates these adjustments into the solution of Problem \ref{problem_viable} presented above.

\textbf{Parameters:} The potential function is defined in the same way as before, but the optimization in \eqref{LTV_guarantees} and the parameter set are modified. Since the state of each subsystem changes in the corresponding state space, the centers of the baseline sets $\bar{c}^x_{i,t}$ and $\bar{c}^u_{i,t}$ in the parametric AG contracts can no longer be set in advance, and are included in the set of parameters.
The extended set of parameters is denoted by $ \alpha^{ext} := \{ \alpha^x_i , \alpha^u_i , \bar{c}_i^x , \bar{c}_i^u \}_{i \in \mathcal{I}}$, where $\alpha^x_i = \alpha^x_{i,1},...,\alpha^x_{i,h-1} , \alpha^u_i = \alpha^u_{i,1},...,\alpha^u_{i,h-1} , \bar{c}_i^x = \bar{c}_{i,1}^x,...,\bar{c}_{i,h-1}^x , \bar{c}_i^u = \bar{c}_{i,1}^u,...,\bar{c}_{i,h-1}^u$. 
The convexity property of the potential function  is not affected by this change in the parameter set. However, the initial time step has no parameters anymore because the initial states are known. Therefore, there is no need to construct parameterized sets for it (i.e., $T^i_0 = 0_{n_i}$ and  $M^i_0 = 0_{m_i}$). Also, since the goal sets, which are the terminal sets, are already known, there is no need to specify any parameterized sets for the last time step. Finally, the potential function is reconstructed as the solution to the following optimization problem:
\begin{subequations} \label{synthesis}
\begin{align}
    & \mathcal{V}_i(\alpha^{ext}) = \underset{\bar{\mathrm{x}}^i_t,T^i_t,\bar{\mathrm{u}}^i_t,M^i_t , d_t^x , d_t^u}{\text{min}}
    \begin{aligned}[t]
      & \sum_{t \in \mathbb{N}_{1,h}} d^x_t + \sum_{t \in \mathbb{N}_{1,h-1}} d^u_t + \omega J_i
    \end{aligned} \notag \\
    &\text{subject to} \notag \\
    & [0_{k-n_i}, D_{i,0}^{aug}] = T^i_1, \label{first_step_generator}\\
    &  \mathcal{Z}( \bar{d}_{i,0}^{aug} , D_{i,0}^{aug})  =  \reduce(\sum_{j \ne i}A_{ij}\bar{\mathrm{x}}^i_0 +  B_{ij}\bar{\mathrm{u}}^i_0 \oplus D_i), \label{first_step_disturbance}\\
    & [A_{ii}T^i_t+ B_{ii} M^i_t , D_{i,t}^{aug} ] =  [T^i_{t+1}], \hspace{2 mm} \forall t \in \mathbb{N}_{1,h-1} \label{eq:mpc_a} \\
    & A_{ii}\bar{\mathrm{x}}^i_t + B_{ii} \bar{\mathrm{u}}^i_t + \bar{d}_{i,t}^{aug} = \bar{\mathrm{x}}^i_{t+1},  \hspace{2 mm} \forall t \in \mathbb{N}_{h-1} \label{eq:mpc_b}\\
    & \mathcal{Z}( \bar{d}_{i,t}^{aug} , D_{i,t}^{aug})  =  \reduce(\bigoplus_{j \ne i}A_{ij}\mathcal{X}_{j,t}(\bar{c}_{i,t}^x , \alpha^x_{j,t}) \notag \\ & \hspace{0.8 cm} \oplus \bigoplus_{j \ne i} B_{ij}\mathcal{U}_{j,t}(\bar{c}_{i,t}^u , \alpha^u_{j,t}) \oplus D_i), \forall t \in \mathbb{N}_{1,h-1} \label{eq:mpc_c}\\
    & \mathcal{Z}(\bar{\mathrm{x}}^i_t,T^i_t) \subseteq \ \mathcal{X}_{i,t}(\bar{c}_{i,t}^x,\alpha^x_{i,t}) \oplus \mathcal{Z}(0,d^x_t I_{n_i}) ,  \forall t \in \mathbb{N}_{1,h-1} \label{eq:mpc_d}\\
    & \mathcal{Z}(\bar{\mathrm{x}}^i_h,T^i_h) \subseteq \ \Omega_i \oplus \mathcal{Z}(0,d^x_h I_{n_i}) , \hspace{2 mm}  \label{eq:mpc_dd}\\
    & \mathcal{Z}(\bar{\mathrm{u}}^i_t,M^i_t) \subseteq \ \mathcal{U}_{i,t}(\bar{c}_{i,t}^u,\alpha^u_{i,t}) \oplus \mathcal{Z}(0,d^u_t I_{m_i}) ,  \forall t \in \mathbb{N}_{1,h-1} \label{eq:mpc_e}\\
    & \bar{\mathrm{x}}^i_0 = x_{i,0}, \label{state_first_step}\\
    & d_t^x , d_t^u \geq 0 , \hspace{2 mm} \forall t \in \mathbb{N}_{1,h},\label{eq:mpc_hausdorffdistance}
\end{align}
\end{subequations}
where constraints \eqref{first_step_generator} and \eqref{first_step_disturbance} are originated from Theorem \ref{Thrm_viable set_single}, while the generators are assigned to zero matrices, where $k$ is the number of columns in $T_1^i$. Also, \eqref{state_first_step} specifies the current state and \eqref{eq:mpc_dd} imposes the terminal condition when $d_h^x=0$, which happens when $\mathcal{V}^*=0$.

\textbf{Objective function:} In our formulation, we capture the goal sets using the set inclusion constraints. However, one might be interested in the optimality of the MPC formulation, instead of only a feasible solution to a goal set. In this case, we wish to minimize the potential function as well as a user-defined cost function $J_i$. The main challenge is to satisfy all the set containment requirements while minimizing $J_i$. A very little protrusion of one of the sets, on the other hand, may cause a cascade of instability among the subsystems.
The objective function is the weighted sum shown in \eqref{synthesis}, where $\omega$ is a scalar weight. Given the user-defined cost function, we start the iterations with $\omega$ set to zero. Following the discovery of correct contracts, where the objective functions with zero weight for all subsystems are zero, the weight $\omega$ is increased as long as $\sum_{t \in \mathbb{N}_{1,h}} d_t^{x^*} + \sum_{t \in \mathbb{N}_{1,h-1}} d_t^{u^*}$ stays zero for all subsystems (i.e. the contracts remain correct), where $ d_t^{x^*}$ and $d_t^{u^*}$ are optimal values for $d_t^{x}$ and $d_t^{u}$, respectively. Since the weight is increased to the point that one of the subsystems disregards correctness, the number of iterations is limited.
\textbf{Communication:}
The method proposed here can be seen as an \emph{iterative} distributed MPC \cite{christofides2013distributed} with disturbance rejection, since neighboring subsystems share information, including parameters and generated gradients, as many times as needed until convergence, at each time step. Following convergence, each subsystem implements the optimal value calculated for $\bar{\mathrm{u}}^i_0$, and the cycle continues. Note that after each convergence, the decentralized output controllers give a guaranteed solution to the goal sets, making our method resilient to connection loss issues. By repeatedly solving the problem, on the other hand, the most up-to-date information is used, which leads to a more optimal solution. In addition, because we may reuse the optimal variables from the previous iteration, there is no need to randomly initialize the parameters each time the MPC is solved. The \textit{warm start} speeds up the convergence process.

\section{Case Studies}
\label{sec_example}
This section contains four case studies. The first is a simple example of finding decentralized RCI sets for a connected LTI system with coupled state constraints. The second example shows the viable sets of a connected network of LTV systems. The goal is to demonstrate sequences of decentralized time-limited viable sets. The third illustrates the scalability of the proposed approach by comparing its computation time with those obtained using a series of benchmark methods. The final example demonstrates the capability of our proposed MPC approach in a power network.
The source code is available on GitHub\footnote{https://github.com/Kasraghasemi/parsi}. A MacBook Pro 2.6 GHz Intel Core i7 with Gurobi \cite{gurobi} as the optimizer was used for the implementations.

\subsection*{Case Study 1}
Consider an LTI system with
\begin{equation*}
A = \left[\setlength{\arraycolsep}{1pt}
  \renewcommand{\arraystretch}{0.6} \begin{array}{c c| c c |c c}
	0.1 & 0.1 & 0.1 & 0.02 & 0.04 & -0.02 \\
    0 & 0.1  & -0.1 & 0.06 & 0 & -0.04 \\
\hline
    -0.08 & 0 & 0.1 & 0.1 & 0.04 & 0.1 \\
    0.02 & -0.06 & 0 & 0.1 & 0.08 & 0 \\
\hline
    0 & 0 & 0.04 & 0.02 & 0.1 & 0.1 \\
    0 & 0 & 0.02 & 0.1 & 0 & 0.1 
\end{array}\right]
\end{equation*} 
as the aggregate matrix. We assume that
the dynamics of the subsystems are only coupled through states, hence $B$ is a block diagonal matrix. The state constraints are expressed in the following coupled form:
\begin{equation*}
    X = \mathcal{Z}(0_6 , 
    \left[\setlength{\arraycolsep}{1pt}
  \renewcommand{\arraystretch}{0.6} \begin{array}{c c| c c |c c}
	1 & 0.1 & 0.2 & 0 & 0 & 0 \\
    0.1 & 1  & 0.02 & 0 & 0 & 0.1 \\
\hline
    0 & 0.01 & 1 & 0 & 0.1 & 0.1 \\
    0.2 & 0 & 0.03 & 1 & 0 & 0.1 \\
\hline
    0 & 0.1 & 0.1 & 0 & 1 & 0.2 \\
    -0.1 & -0.02 & 0.1 & 0 & 0 & 1 
\end{array}\right]).
\end{equation*}
The subsystems are defined by \eqref{subsystems_LTV} with $
     B_{ii} = \begin{bmatrix} 
0 \\
0.1
\end{bmatrix},
U_i = \mathcal{Z}(0_1 , [1]) ,
D_i = \mathcal{Z}(0_2 ,     \begin{bmatrix} 
0.3 & 0 \\
0 & 0.3
\end{bmatrix})$.
Our goal is to identify decentralized RCI sets $\Omega_1,\Omega_2, \Omega_3$ by breaking the 6-dimensional system into three 2-dimensional subsystems. The first step is to decouple the state admissible set $X$ into $X_1,X_2$, and $X_3$. We accomplish this by using the semi-definite program in \eqref{decompose_optimization}. The outcomes are illustrated in Fig \ref{decomposition_states_fig}. The decentralized RCI sets and the correct set of parameters are determined using the compositional method provided in Sec. \ref{compos_method}. 
Fig \ref{valid_alpha_fig} shows the projection of potential function and the correct set of parameters on $\alpha^x_1[1]-\alpha^x_1[2]$ space, which are two of the parameters. 

\begin{figure} 
\begin{center}
\includegraphics[height=3.0cm]{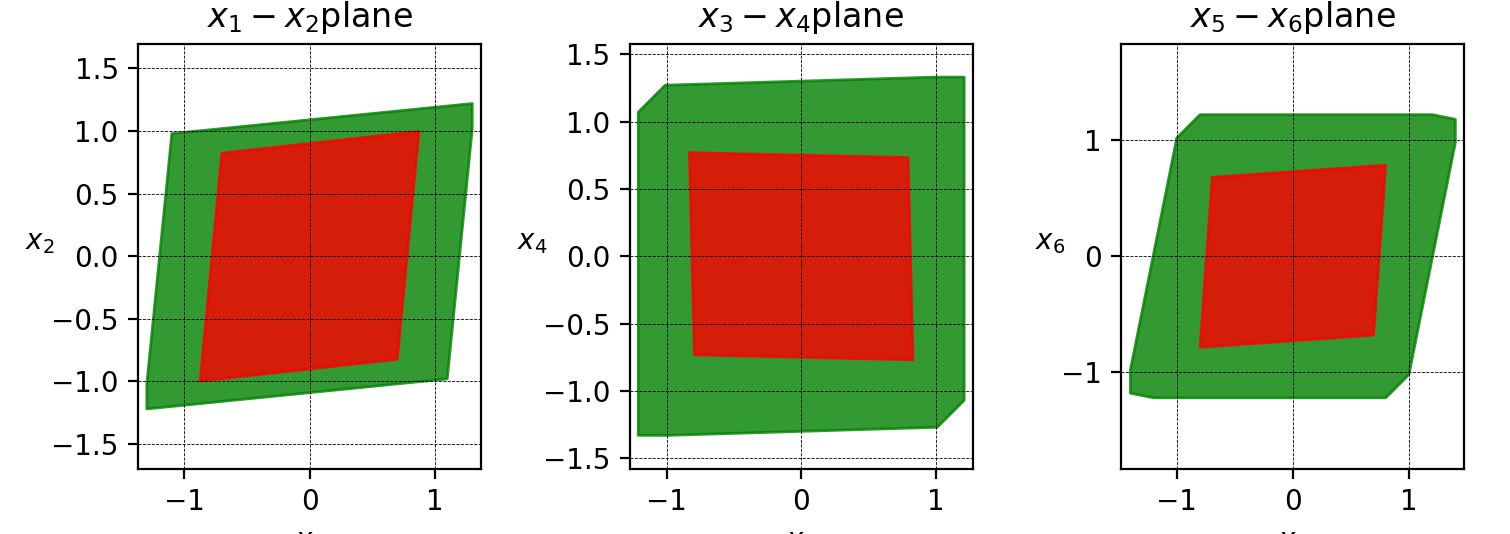}    
\caption{Case Study 1: The projections of the 6-dimensional admissible sets to the corresponding plane for each subsystem are shown in green. The decomposed sets are depicted in red. As expected, all the red sets are subsets of the green sets.}  
\label{decomposition_states_fig}                                 
\end{center}                                 
\end{figure}
\begin{figure}[t] 
  \centering
  \includegraphics[width=0.24\textwidth]{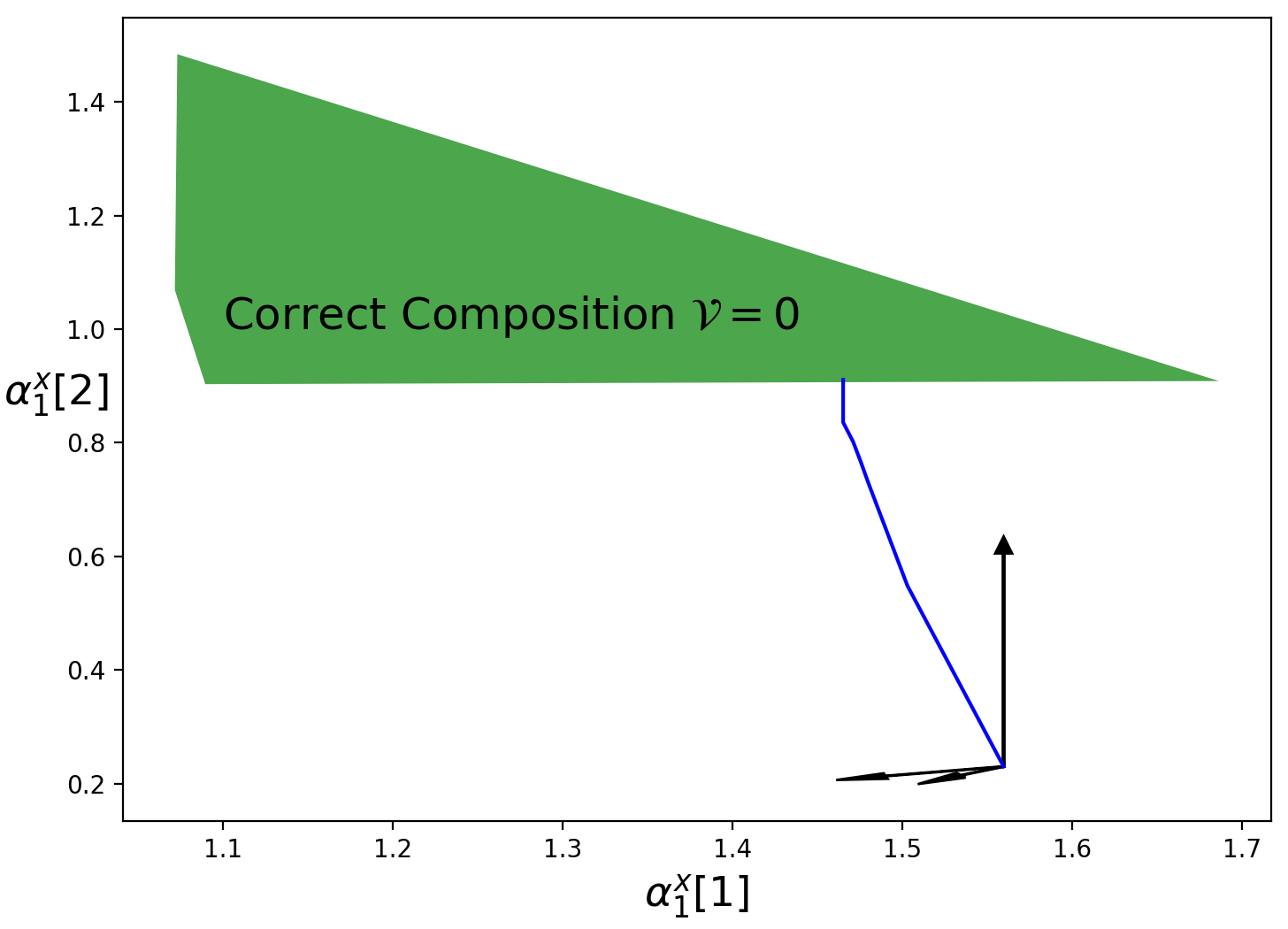}
  \includegraphics[width=0.23\textwidth]{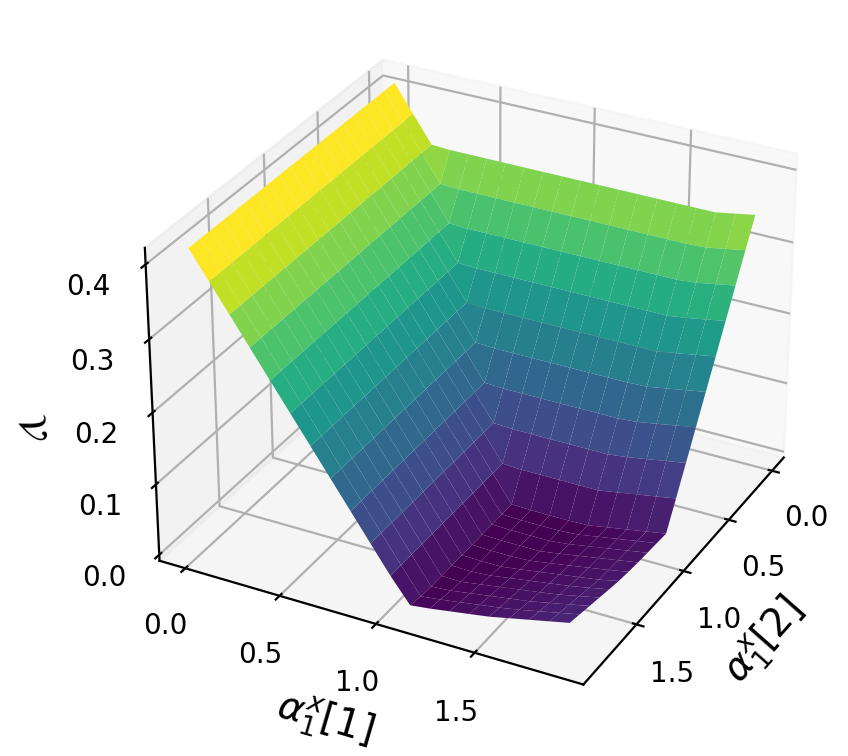}
  \caption{Case Study 1: [Left] The green polytope shows the projection of the correct set of parameters onto the plane of two of the parameters. The blue trajectory represents the path obtained using the compositional method to get to the final correct parameters, starting from random parameters. The three arrows show three preferred directions ($\nabla_\alpha \nu_i$) for each of the three subsystems at the start point. [Right] The potential function as a function of the same two parameters, with all other parameters set at their final correct values. Its level set at zero corresponds to the correct set of parameters. The potential function is a piece-wise affine function.}
  \label{valid_alpha_fig}
\end{figure}

\subsection*{Case Study 2}
Our goal in this case study is to determine the sequences of decentralized limited-time viable sets for a 8-dimensional LTV system with 
\begin{equation*}
    0.01 \left[ \setlength{\arraycolsep}{1pt}
  \renewcommand{\arraystretch}{0.6} \begin{array}{c c| c c |c c | c c}
t  & 100  & 0.1 & 0.1 & 0 & 0 & -\sin(t) & t \\
0 & 100 & 0.1 & 0.1 & 0 & 0 & \log(t+1) & \cos(t)\\
\hline
0.1 & 0.1 & 100 & 100 & 0.1 & 0.1 & 0 & 0 \\
0.1 & 0.1 & 0 & 100 & 0.1 & 0.1 & 0 & -1 \\
\hline
0 & 0 & 0.1 & 0.1 & 100 & 100 & 0 & 1 \\
0.1t^2 & 0 & 0.1 & 0.1 & 0 & 100 & 0 & 0 \\
\hline
t & 0 & 0 & 0 & 0 & 0 & 100 & 100 \\
-t & 1 & 0 & 1 & 0 & 0 & 0 & t 
\end{array}\right],
\end{equation*}
as the aggregate time-variant $A_t$ matrix.
There are four subsystems in the form of \eqref{subsystems_LTV} that are only coupled through states in their dynamics. We assume $B_{ii,t} = [0 , 0.1]^T,
D_{i,t} = \mathcal{Z}(0_2 , \diag([0.4 , 0.4])$ ,$ X_{1,t} = \mathcal{Z}(0_2 ,  \diag([ 5 -  \sin\dfrac{\pi t }{15} , 6 - 5.5 \sin \dfrac{\pi t}{12} ])$ ,$X_{2,t} = \mathcal{Z}(0_2 ,  \diag([ 5 -  2 \sin\dfrac{\pi t }{8} , 6 - 5.5 \sin \dfrac{\pi t}{20} ])$ ,$X_{3,t} = \mathcal{Z}(0_2 ,  \diag([ 5 -   \cos\dfrac{\pi t }{15} , 6 - 5.5 \cos \dfrac{\pi t}{12} ])$ ,$X_{4,t} = \mathcal{Z}(0_2 ,  \diag([ 5-\dfrac{t }{5} ,5 -  \dfrac{t }{5}  ]) , U_i = \mathcal{Z}(0,[10])$
The horizon is $15$ ($t \in \mathbb{N}_{15}$).
The results, which are displayed in Fig \ref{time_limited_viable_set_fig}, were found using the centralized method described in Sec. \ref{centralized_approach}. The sizes of the viable sets tend to increase in time due to the additive disturbances. Furthermore, since the objective function is designed to reduce the area of the viable sets, the viable sets at the first time step are points, and there are no sets for the first time step in Fig. \ref{time_limited_viable_set_fig}.
\begin{figure}[t]
  \centering
  \includegraphics[scale=0.65]{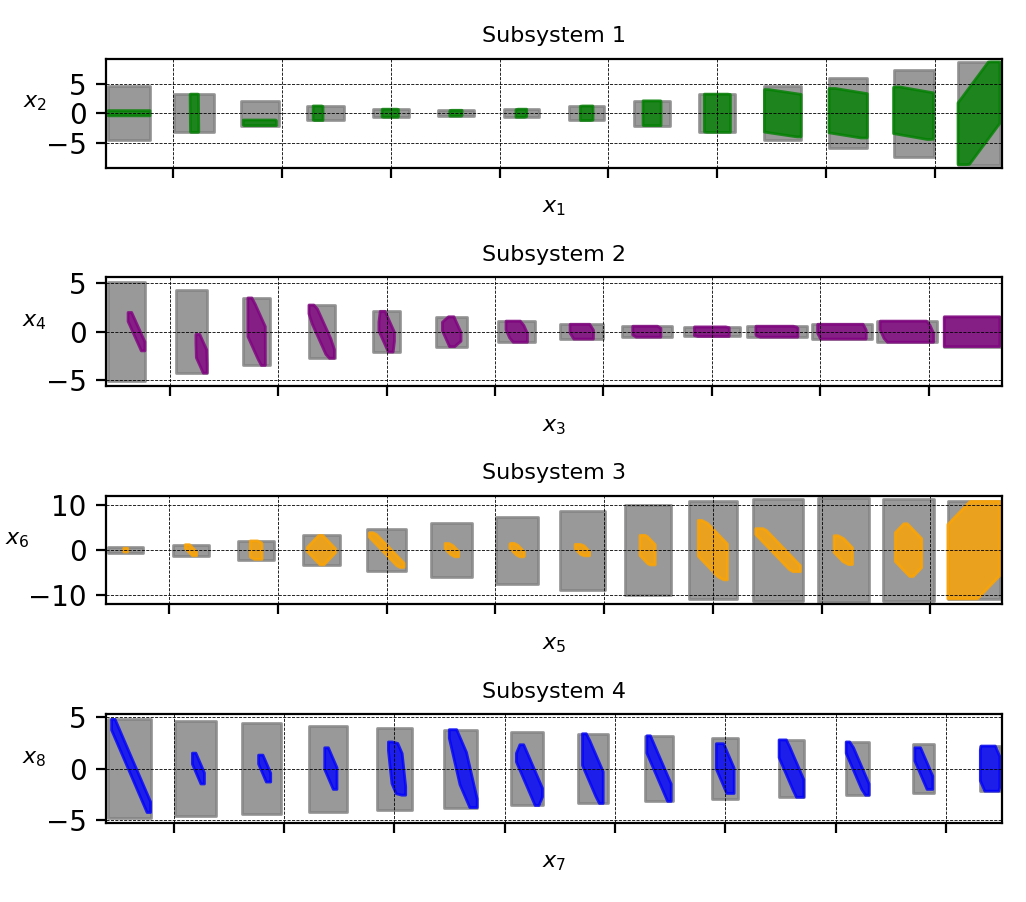}
  \caption{The decentralized viable sets for the four coupled LTV subsystems from Case Study 2. To avoid sets from overlapping in each subsystem, each set is shifted to the right of the preceding set. In fact, they all have a common center. The gray areas are state constraints $X_{i,t}$, which correctly contain the viable sets entirely.}
  \label{time_limited_viable_set_fig}
\end{figure}

\subsection*{Case Study 3}
This example is adopted from \cite{Motee2008}, which creates a vast random network of connected LTI subsystems by randomly spreading a predetermined number of points in a $100$ by $100$ square area. Each point represents a subsystem with $
    A_{ii} = \begin{bmatrix} 
    1 & 0.2\\
    0 & 1
    \end{bmatrix} , 
B_{ii} = \begin{bmatrix} 
    0\\
    0.2
    \end{bmatrix},
    X_i = \mathcal{Z}(0, 5I_2),
U_i = \mathcal{Z}(0,5I_1),
D_i = \mathcal{Z}(0, 0.1I_2).
$
If the Euclidean distance $\dist(i,j)$ between two points representing two subsystems $i$ and $j$ is less than $10$, they are dynamically coupled with the following $A_{ij} =
    \dfrac{\lambda}{1 + \dist(i,j)}
    \begin{bmatrix} 
    1 & 1\\
    1 & 1
    \end{bmatrix}$, where $\lambda$ is a constant scalar that limits the effect of couplings when the number of subsystems is large.
We compare three approaches to obtain the RCI sets: (i) \textbf{Cen} uses Theorem \ref{thrm_single_RCI} to construct a centralized RCI set as an output of centralized optimization problem, (ii) \textbf{DecCen} uses the centralized optimization presented in Sec. \ref{centralized_approach}, with a set of decentralized RCI sets as the result, and (iii) \textbf{Compose} uses the compositional method described in Sec. \ref{compos_method}, which likewise produces decentralized RCI sets. Note that in the LTI case, we need to apply the approach from Remark \ref{rmrk_RCI} to make all constraints linear.

Table~\ref{tabel1} shows the average synthesis times  for the three approaches and the average number of iterations for the compositional method for 10 sample runs for systems with increasing state dimensions $n$. The reported times include only the time spent solving the optimization problems, not the time spent constructing them. Approach (i) failed to handle large dimensions and timed out. Given the polynomial computational complexity of linear programs, this was expected. While having a larger number of variables, approach (ii) provides better results, which is mainly due to the sparsity in the control matrices. The sparsity originates from the fact that we are considering separable viable and action sets for different subsystems. The second strategy is the quickest of the three in smaller dimensions. Approach (iii), on the other hand, demonstrates its scalability by expanding to massive dimensions. The synthesis time could be reduced further by utilizing parallel computation techniques using multi-core computers, which we will address in future work.
\label{case_table}
\begin{table}[t]
\caption{Average synthesis times (in seconds) for the three approaches from Case Study 3 and the average number of iterations for the compositional method}
\resizebox{0.48\textwidth}{!}{
\begin{tabular}{|c|c|c|c||c|c|}
\hline
$n$ & $\lambda$ & \texttt{Cen} & \texttt{DecCen} & \texttt{Compose} & Avg. Iter.
\\ \hline
10    & 0.1       & 0.26     & 0.01        & 0.03       & 3.4       \\ \hline
20    & 0.1       & 3.23     & 0.10        & 0.05       & 38.3      \\ \hline
30    & 0.1       & 19.60    & 0.26        & 0.24       & 99.0      \\ \hline
50    & 0.01      & 193.69   & 1.03        & 0.56       & 24.8      \\ \hline
100   & 0.01      & ---      & 4.34        & 1.66       & 33.8      \\ \hline
200   & 0.01      & ---      & 37.03       & 6.46       & 73.6      \\ \hline
500   & 0.01      & ---      & 576.72      & 10.63      & 47.0      \\ \hline
1000  & 0.001     & ---      & ---         & 12.40      & 30.3      \\ \hline
10000 & 0.0001    & ---      & ---         & 153.00     & 30.0      \\ \hline
\end{tabular} 
}
\label{tabel1}
\end{table}

\subsection*{Case Study 4}

The load-frequency control (LFC) problem in power networks \cite{Camacho2015} is addressed in this example. A power network is a collection of areas, where each area has its own generator and users, and can transfer its excess power to other areas through a set of links. The phase angle $\delta_{i,t}$ and frequency $f_{i,t}$ comprise the state of area/subsystem $i$ at time $t$. The control input is the amount of power $u_{i,t}$ produced by the generator. The goals are to keep the system in a state that is close to its nominal value and to keep the power exchanges between areas on schedule. Since only the departure from the nominal state matters, the state vector for subsystem $i$ is written as $x_{i,t} = [\Delta \delta_{i,t} , \Delta f_{i,t}]^T$, with $\Delta$ denoting the divergence from the nominal value. The system dynamics from \cite{Camacho2015} is discretized using the Euler method, resulting in $
    A_{ii} = \begin{bmatrix} 
    1, & 2 \pi \Delta t\\
    \dfrac{- \Delta t k_{p_i}}{2 \pi T_{p_i}}  \displaystyle\sum_{j \in \mathcal{N}_i}K_{s_{ij}},  & 1 - \dfrac{\Delta t}{T_{p_i}}
    \end{bmatrix} ,
    A_{ij} = \begin{bmatrix} 
    0, & 0 \\
    \dfrac{ \Delta t k_{p_i} K_{s_{ij}}} {2 \pi T_{p_i}} ,  & 0
    \end{bmatrix} ,
    B_{ii} = \begin{bmatrix} 
    0 \\
    \dfrac{K_{p_i} \Delta t}{ T_{p_i}}
    \end{bmatrix}$,
where $\mathcal{N}_i$ is a set containing the neighbours of subsystem $i$ and $\Delta t , K_{p_i} , K_{s_{ij}}, T_{p_i}$ are the discretized time interval, system gain, synchronizing coefficient between area $i$ and $j$, and system model time constant, respectively, which are set to  0.1s, 110, 0.5, and 25s for all areas. Consider a network of four fully connected areas with bidirectional connections between each pair, with the following bounds over control input and load disturbance for each subsystem:
$|u_{i,t}| \leq 0.1$ and $D_i = \mathcal{Z}(0_2 ,\diag([ \epsilon , \dfrac{-\Delta t K_{p_i} \Delta P_{d_i}}{T_{p_i}} ]))$, where $\epsilon = 10^{-10}$ and $\Delta P_{d_i}=0.01p.u$ is the load disturbance for area $i$.
The initial state for the aggregate system is $[-0.2 , 0.1 , 0.02 , 0.01 , 0.1 , -0.05 , -0.03 , -0.01]^T$ and the goal is to go to the set where $|\Delta \delta_{i,h}| \leq  0.01 , |f_{i,h}| \leq  0.01 , \forall i$.
Using the solution to Problem \ref{Problem_RCI}, we first find a set of decentralized RCI sets inside the goal set. Then, we proceed with the strategy from Sec. \ref{section:mpc} with horizon $h=5$. Also, in the cost function, the $R$ and $Q$ matrices are set to zero and the identity matrix, respectively. The output viable sets after solving the problem are illustrated in green for each subsystem separately in Fig. \ref{power_system_mpc_fig}, where as expected, the final viable set in each subsystem ends up inside its goal set. 
If the goal is just to reach the goal set, there is actually no need to solve the problem again in the next time step for the new initial states. Each subsystem can implement the controller in the derived sequence, which by construction. has been shown it maintains the states inside the derived viable sets.
\begin{figure}[t]
  \centering
  \includegraphics[scale=0.6]{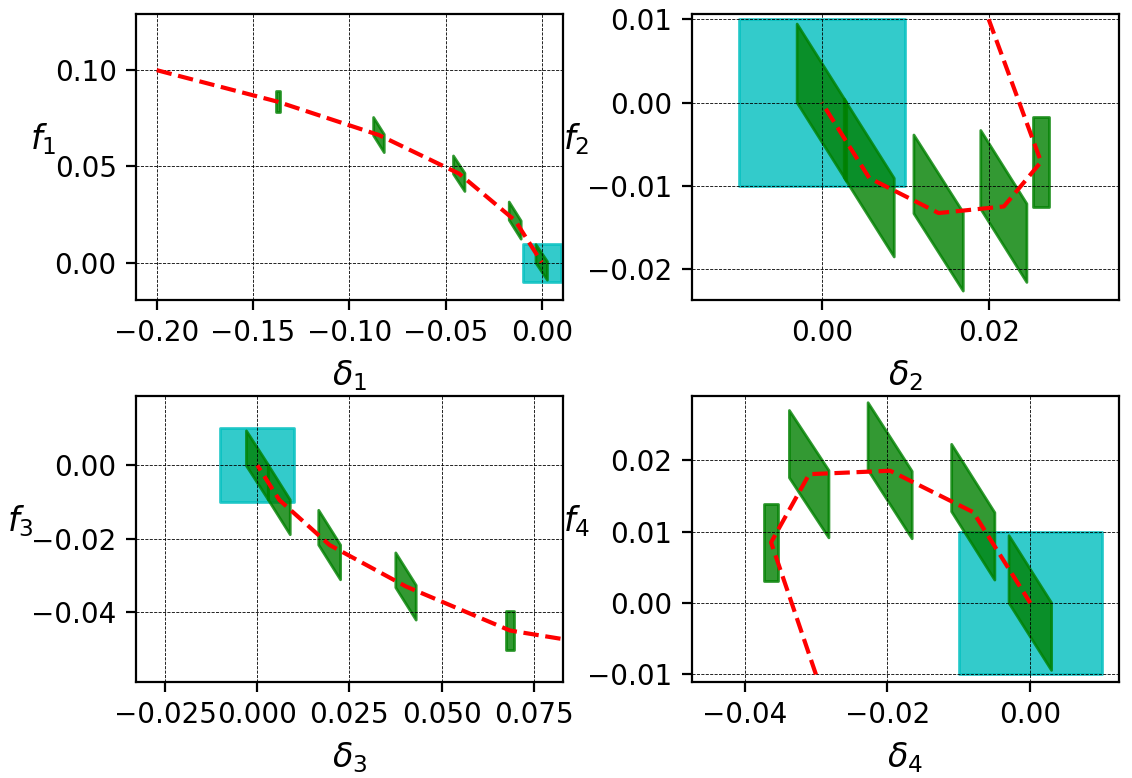}
  \caption{The goal sets and sequences of viable sets for each area of the power network from Case Study 4 are shown in blue and green, respectively. The red dashed lines begin at the initial state and link the centers of the viable sets, i.e., the nominal predicted trajectory.}
  \label{power_system_mpc_fig}
\end{figure}

\section{Conclusion and Future Work}

We proposed a convex parameterization of assume-guarantee contracts that facilitates compositional control synthesis of large-scale linear systems with uncertainties. We achieved linear time complexity, while other state-of-the-art algorithms offer polynomial time complexity at best. Extensions to nonlinear and (partially) unknown dynamics, investigating probabilistic assume-guarantee contracts, and implementations in robotic experimental platforms are directions of future research.

\begin{ack}                               %
This work was partially supported by the National Science Foundation under grants 
IIS-2024606 and IIS-1723995. 
\end{ack}

\bibliographystyle{plain}        
\bibliography{references}           
\section*{Appendix}
{\bf Proof of Lemma \ref{lemma_weighted}}
Using constraints (\ref{zon_containment}) for $\mathcal{Z}(\bar{c}_1,G_1)  \subseteq  \mathcal{Z}(\bar{c}_2, G_2\diag(\alpha))$, we have:
\begin{equation}
\begin{array}{c}
G_1 = G_2 \diag(\alpha) \Gamma , \\
c_2 - c_1 = G_2 \diag(\alpha) \gamma ,   || [\Gamma,\gamma] ||_\infty \leq 1.
\end{array}
\end{equation}
We can replace $\diag(\alpha) \Gamma$ and $\diag(\alpha) \gamma$ with $\Gamma ^{new}$ and $\gamma^{new}$, respectively. We have:
\begin{subequations}
\begin{equation}
G_1 = G_2 \Gamma^{new},
\end{equation}
\begin{equation}
c_2 - c_1 = G_2 \gamma^{new},
\end{equation}
\begin{equation} \label{zon_con_mod}
|| [\diag(\alpha^{-1}), \Gamma^{new},\diag(\alpha^{-1})\gamma^{new}] ||_\infty \leq 1,
\end{equation}
\end{subequations}
where $\alpha^{-1}$ is element-wise. In  (\ref{zon_con_mod}), each row of matrix $[\Gamma^{new},\gamma^{new}]$ is divided by the corresponding element in vector $\alpha$. Because all the elements of $\alpha$ are positive, we can multiply the inequality by $\diag(\alpha)$ and have:
\begin{equation}
    [|\Gamma^{new}|,|\gamma^{new}|]  \mathbb{1}_s \leq \alpha
\end{equation}
\section*{Solution to Subproblem \ref{subproblem}}
By definition, we have $\{ [x_1^T, x_2^T, ..., x_\eta^T]^T | \forall x_i \in X_i , i\in \mathcal{I} \} =  \prod_{i \in \mathcal{I}}X_i$. Thus, it can be seen that  for all $x_i\in X_i, [x_1^T, x_2^T, ..., x_\eta^T]^T \in X$ holds, if
\begin{equation} \label{decomposition_sufficient_condition}
    \prod_{i \in \mathcal{I}}X_i \subseteq X, 
\end{equation}
where $X_i = \mathcal{Z}(c_i,G_i)$ and
\begin{equation*}
    \prod_{i\in \mathcal{I}} X_i = \mathcal{Z}([c_1^T, \cdots , c_\eta^T]^T,\text{BlockDiag}([G^x_1, \cdots, G^x_\eta])).
\end{equation*}
Since $c_i$s and $G_i^x$s are unknown. we encode \eqref{decomposition_sufficient_condition} into a linear program with the help of the linear containment encoding proposed in \cite{sadraddini2019linear}. However, we are interested in finding the maximum volume for $\prod_{i \in \mathcal{I}}X_i $.


\begin{assumption} \label{G_square}
Each $G^x_i, i\in \mathcal{I}$ is a symmetric square positive semi-definite matrix.
\end{assumption}

The assumption that the matrices are square makes the calculation of the volume of $\prod_{i\in \mathcal{I}} X_i$ easier by the formulation in \eqref{zonotope_volume}. However, the function in \eqref{zonotope_volume} is still non-concave, which is not desirable for an optimization. It can be shown that by further assuming that all $G_i^x$ are symmetric positive semi-definite $((G^x_i)^T=G^x_i \text{ and } G^x_i \succeq 0)$, the logarithm of the volume in \eqref{zonotope_volume}, which is equal to
\begin{multline}
    \text{log} ( \text{Volume} ( \prod_{i\in \mathcal{I}} X_i )) = \\ n\text{log}(2) + 2\text{log} ( \text{det} ( \text{BlockDiag}([G^x_1, \cdots, G^x_\eta]) ) )
\end{multline}
becomes concave because $log(det(G))$ when $G$ is a positive semi-definite matrix is a concave function. Considering the assumptions and removing the constants in the objective function, the final optimization problem for decomposing $X$ is as follows:
\begin{subequations} \label{decompose_optimization}
\begin{align}
    \min _{c_i,G^x_i } \quad & \text{log} ( \text{det} ( \text{Blk}([G^x_1, \cdots, G^x_\eta]) ) ) \\
    \textrm{s.t.} \quad &\prod_{i \in \mathcal{I}}{\mathcal{Z}(c_i,G^x_i)} \subseteq X, \hspace{2 mm} \label{eq:subset} \\
    & \text{Blk}(G_1^x, ..., G_\eta^x) \succeq 0, \\
    & G_i^x \text{ is a symmetric square matrix} , \forall i \in \mathcal{I}
\end{align}
\end{subequations}
This can be solved with any semi-definite programming solver. In this paper, we solved it by SCS \cite{SCS} using drake \cite{drake}.
                            
\end{document}